\documentclass[10pt]{amsart}
\usepackage{amsmath,amsfonts,amssymb,mathrsfs}
\usepackage{float}
\usepackage{epsfig,subfigure}

\title[The Vlasov-Born-Infeld system]
{Low dimensional Born-Infeld equations coupled with a collisionless matter model}

\author{Ho Lee}
\address{Department of Mathematics, Kyung Hee University, Seoul, 130-701, Republic of Korea}
\email{holee@khu.ac.kr}

\begin{document}

\newtheorem{theorem}{Theorem}[section]
\newtheorem{lemma}{Lemma}[section]
\newtheorem{corollary}{Corollary}[section]
\newtheorem{proposition}{Proposition}[section]
\newtheorem{remark}{Remark}[section]
\newtheorem{definition}{Definition}[section]

\renewcommand{\theequation}{\thesection.\arabic{equation}}
\renewcommand{\thetheorem}{\thesection.\arabic{theorem}}
\renewcommand{\thelemma}{\thesection.\arabic{lemma}}
\newcommand{\bbr}{\mathbb R}
\def\charf {\mbox{{\text 1}\kern-.24em {\text l}}}

\keywords{Vlasov; Born-Infeld; one and one-half dimensions}
\subjclass{35L60; 35Q83}

\maketitle

\begin{abstract} We consider the Born-Infeld nonlinear electromagnetic field equations
and study its Cauchy problem in the case that the Vlasov equation is considered
as a matter model. In the present paper, the Vlasov equation is considered on the so-called
one and one-half dimensional phase space,
and in consequence the Born-Infeld equations are reduced to a quasilinear
hyperbolic system with two unknowns.
A transformation is introduced in order to make the field equations
easy to handle, and suitable assumptions are made on initial data
so that the nonlinearity of the field is controlled.
\end{abstract}

%
%
\section{Introduction}
The Born-Infeld (BI) electromagnetic theory \cite{BI34} was originally proposed as a nonlinear
correction of the Maxwell theory in order to
overcome the problem of infinities in the classical electrodynamics of point particles.
The underlying idea was to simply modify the classical theory not to have physical quantities of infinities,
that is the principle of finiteness.
It was to replace the original Lagrangian density for the Maxwell electrodynamics
with a square root form
with a parameter $b$, by which the finiteness of electric fields is ensured.
This was
the same with the way the special relativity has taken into account the finiteness of the speed of light $c$,
i.e., $-mc^2\sqrt{1-v^2/c^2}$ replaced $\frac{1}{2}mv^2$.
The exact form of the BI Lagrangian density appears in the next section.

This theory in recent years again has received much attention since the string theory found its relevance to
the BI theory \cite{Po,Zw}: the dynamics of electromagnetic fields on D-branes is described by
the BI theory, and finiteness of electric fields is naturally observed.
Be that as it may, this paper is not going to the string theory, but the BI theory
will be considered rather as a nonlinear version of the Maxwell theory
as it was the original standpoint of Born and Infeld.
Hence, it makes sense to choose the Vlasov equation as a matter model since
this equation describes the dynamics of classical particles.

In the present paper, we study the BI equations coupled with the Vlasov equation
as a matter model, and this induces inhomogeneities to the BI equations.
For homogeneous cases, we can find several results on the Cauchy problem.
In \cite{CH03}, Chae and Huh proved the existence of global classical solutions for small initial data.
They considered the BI equations as quasilinear wave equations and crucially used the null form structure
of the nonlinear terms. On the other hand, Brenier \cite{Br04} considered the BI equations
as a system of hyperbolic conservation laws.
He supplemented additional conservation laws to obtain an augmented BI system
and discussed conditions for existence of global solutions in one dimensional case.
The similar frameworks can be found in \cite{Pe06,Pe07,PR07,Se,Se04}, and for different approaches
we refer to \cite{FOP02,Gi98,Ya00}.
In this paper, the BI equations will be considered as a system of hyperbolic equations as in \cite{Br04},
but it differs in that inhomogeneous terms are taken into account.

On the other hand, the Vlasov equation has been widely used to describe matters in connection with
other field equations: electrostatic or electromagnetic fields, non-relativistic or fully relativistic
gravitational fields, and so on. For detailed discussions on the Vlasov-type equations,
we refer to \cite{An02,Gl} and the references therein.
The analysis of the Vlasov-Maxwell system is applied to our case,
and as a simple but nontrivial case we consider the ``one and one-half dimensional'' case as in \cite{GSc90}.
The Vlasov equation in this dimensional case is introduced in Section 2.

Before we proceed, we briefly discuss some points of this paper.

1. To make a coupled system of the BI and the Vlasov equations, we should determine an analogous equation
corresponding to the Lorentz equation in the Maxwell case. In other words, we have to
derive an equation of motions of particles interacting with the BI electromagnetic field.
However, as it was pointed out in \cite{Chr98},
it is impossible to derive separate equations for the charged particles due to the nonlinearity of the field.
Instead, we use the ``generalized Lorentz force'',
which is available for small and almost constant fields \cite{Che99,Che05},
and this will be discussed in Section 2.

2. As we mentioned above, the BI equations will be considered as a system of hyperbolic equations.
To better understand its hyperbolic structure, we introduce a transformation in Section 3.
This transformation is quite tricky but simple. By this transformation,
the linear degeneracy of the BI system is then easily seen, and moreover it makes the system easy to handle,
for instance the proof of Lemma 5.6 would be much more complicated without the transformation.

3. The main difference from the Maxwell case is that we have to control
the characteristics of the field equations as well as of the Vlasov equation.
In Glassey-Strauss' result \cite{GSt86}, spatial and temporal derivatives were split
into linear combinations of the Maxwell and the Vlasov characteristics called $T$ and $S$ derivatives respectively.
This was possible because
the Maxwell field always propagates with the speed of light, while the Vlasov particles do not move
with the speed of light.
However, it will be shown in Section 3 that the BI field
propagates with speed $-\cos\beta$ or $\cos\alpha$, which may have values less than the
speed of light $c=1$. Hence, we should make a suitable assumption on initial data
in order to apply the argument of \cite{GSt86}, and this will be discussed in Section 4.

4. In contrast to Glassey-Schaeffer's result \cite{GSc90}, we obtain a local-in-time result.
In that paper, the Maxwell electromagnetic field is estimated as $|E|+|B|\leq c <\infty$,
(see Lemma 1 and Corollary 1 in \cite{GSc90}),
and this estimate implies that any particles considered have momentum $v$ with growth order $v\sim t$.
However, this is not enough to be applied to our case.
If a momentum grows, and consequently tends to infinity,
then its corresponding velocity will tend to the speed of light,
and finally a resonance between the particle velocity and the BI field propagation
may appear in a finite time.
This will be discussed in Section 5.2.
It seems that more qualitative analysis for
the BI and the Vlasov characteristics
is required to obtain a global-in-time result.\bigskip

\noindent{\bf Notations.} We collect some notations which are used in this paper.
\begin{itemize}
\item The speed of light $c$, and mass $m$ and charge $q$ of each particle are assumed to be unity:
\[
c=m=q=1.
\]
\item Greek indices run from $0$ to $3$ and Latin indices from $1$ to $3$.
The indices are raised or lowered by multiplication with the Minkowski metric
$\eta_{\alpha\beta}=\mbox{diag}(-1,1,1,1)$, and we use the Einstein summation convention such as
\[
{\tt x}_\alpha {\tt y}^\alpha={\tt x}^\alpha {\tt y}_\alpha
=\eta_{\alpha\beta}{\tt x}^\alpha{\tt y}^\beta
=-{\tt x}^0{\tt y}^0+{\tt x}^1{\tt y}^1+{\tt x}^2{\tt y}^2+{\tt x}^3{\tt y}^3.
\]
This notation will appear only in Section 2.1.
\item For a scalar or vector valued function ${\tt v}\in\bbr^d$, $d=1,2,3$, we define
\[
\hat{{\tt v}}:=\frac{{\tt v}}{\sqrt{1+|{\tt v}|^2}}.
\]
Then for a momentum $v\in\bbr^2$, its corresponding velocity is defined by $\hat{v}$.
\item The usual $L^p$ norms, $1\leq p\leq\infty$, are used. For a function ${\tt f}={\tt f}({\tt x})$ defined on
\[
{\tt x}=({\tt x}_1,\cdots,{\tt x}_r)\in\bbr^{d_1}\times\cdots\times\bbr^{d_r},
\]
its norms are defined as follows: for $1\leq p<\infty$ and each $i=1,\cdots,r$,
\begin{align*}
||{\tt f}||_p^p&=\int_{\bbr^{d_1}\times\cdots\times\bbr^{d_r}}|{\tt f}({\tt x})|^p\,
d{\tt x}_1\, \cdots\, d{\tt x}_r,\\
||{\tt f}({\tt x}_i)||_p^p&=\int_{\bbr^{d_1}\times\cdots\times\bbr^{d_{i-1}}\times
\bbr^{d_{i+1}}\times\cdots\times\bbr^{d_r}}|{\tt f}({\tt x})|^p\,
d{\tt x}_1\, \cdots\, d{\tt x}_{i-1}\,d{\tt x}_{i+1}\, \cdots\, d{\tt x}_r,
\end{align*}
and similarly for $p=\infty$ and each $i=1,\cdots,r$,
\begin{align*}
||{\tt f}||_\infty&=\sup\{|{\tt f}({\tt y})|:
{\tt y}\in\bbr^{d_1}\times\cdots\times\bbr^{d_r}\},\\
||{\tt f}({\tt x}_i)||_\infty&=\sup\{|{\tt f}({\tt y})|:
{\tt y}\in\bbr^{d_1}\times\cdots\times\bbr^{d_r},{\tt y}_i={\tt x}_i\},
\end{align*}
and so on.
\end{itemize}

\bigskip

This paper is organized as follows: we first set up the Cauchy problem for the coupled system in Section 2,
where the Born-Infeld equations will be properly coupled to the Vlasov equation in the one and one-half
dimensional case. In Section 3, we introduce a transformation which transforms the Born-Infeld equations
into a quasilinear system in a diagonal form. In Section 4, we state the main result of the present paper.
Suitable assumptions will be made on initial data in this section.
Section 5 and 6 are devoted to the proof of the main theorem.
%
%
\section{Problem setting: coupled system in one and one-half dimensions}\setcounter{equation}{0}
In this section, we study the Born-Infeld equations and the Vlasov equation to consider
their coupled system in a low dimensional case which is
called the one and one-half dimensions. In the first part we briefly review the Born-Infeld electromagnetic theory,
and then in the second part we consider the one and one-half dimensional
case for the coupled system.
%
%
\subsection{The Born-Infeld equations}
One can obtain field equations by constructing a suitable Lagrangian density $\mathcal{L}(F)$ and
then applying the Euler-Lagrange equations to it.
In the case that charged particles are given as source,
the interaction between the particles and field must be considered in
the construction of a Lagrangian density.
In this paper, we simply construct a Lagrangian density
by adding the standard interaction term $A_\alpha j^\alpha$
to the Lagrangian as follows:
\[
\mathcal{L}^{total}=\mathcal{L}(F)+A_\alpha j^\alpha,
\]
where $\mathcal{L}(F)$ is the Lagrangian density for the field itself,
and $A_\alpha$ is a four-potential from which the electromagnetic field tensor is defined by
$F=F_{\alpha\beta}=\partial_\alpha A_\beta-\partial_\beta A_\alpha$.
In the Born-Infeld theory, the Lagrangian density $\mathcal{L}(F)$ is defined as follows:
\[
\mathcal{L}(F)=b^2\left(1-\sqrt{-\det\left(\eta_{\alpha\beta}+\frac{1}{b}F_{\alpha\beta}\right)}\right),
\]
where $\eta_{\alpha\beta}=\mbox{diag}(-1,1,1,1)$ is the Minkowski metric, and $b>0$ is a parameter
which measures the nonlinearity of the field.
For simplicity, we take $b=1$ in the present paper. The Euler-Lagrange equations now give the equations of field:
\[
\partial_\beta G^{\alpha\beta}=j^\alpha,\quad\mbox{where}\quad G^{\alpha\beta}=-\frac{\partial\mathcal{L}(F)}{\partial F_{\alpha\beta}}.
\]
To have explicit formulae for the field, instead of the tensor form, we use $E$, $B$, $D$, and $H$ in place of $F_{\alpha\beta}$ and
$G_{\alpha\beta}$ as follows:
\[
F_{i0}=E_i,\quad F_{ij}=\sum_{k=1}^3\epsilon_{ijk}B_k,\quad
G_{i0}=D_i,\quad\mbox{and}\quad G_{ij}=\sum_{k=1}^3\epsilon_{ijk}H_k,
\]
where $\epsilon_{ijk}$ is the completely antisymmetric tensor such that $\epsilon_{123}=1$, $\epsilon_{213}=-1$, and so on.
Then, we get the following system of PDEs:
\begin{equation}\label{2.1}
\begin{aligned}
&\nabla\times E=-\partial_tB,\quad \nabla\cdot B=0,\\
&\nabla\cdot D=\rho,\quad \nabla\times H=\partial_t D+j,
\end{aligned}
\end{equation}
where $\rho=j^0$ and $j=(j^1,j^2,j^3)$ are the charge density and the current density respectively.
Note that we obtain Maxwell's equations by setting $D=E$ and $H=B$.
The main difference from the Maxwell theory is the following nonlinear constitutive
relations between $E$, $B$, $D$, and $H$:
\[
\begin{aligned}
D&=\frac{\partial\mathcal{L}}{\partial E}=\frac{E+(E\cdot B)B}{\sqrt{1-|E|^2+|B|^2-(E\cdot B)^2}},\\
H&=-\frac{\partial\mathcal{L}}{\partial B}=\frac{B-(E\cdot B)E}{\sqrt{1-|E|^2+|B|^2-(E\cdot B)^2}}.
\end{aligned}
\]
In the above relations, $D$ and $H$ are given as functions of $E$ and $B$. However,
we convert them into a more convenient form for later use, that is,
we rewrite $E$ and $H$ as functions of $D$ and $B$ as follows:
\begin{align}
E&=\frac{D-(D\times B)\times B}{\sqrt{1+|D|^2+|B|^2+|D\times B|^2}},\label{2.2}\\
H&=\frac{B+(D\times B)\times D}{\sqrt{1+|D|^2+|B|^2+|D\times B|^2}},\label{2.3}
\end{align}
which can be verified by direct calculations (see Chapter 20 of \cite{Zw}). As a result, we obtain the field equations \eqref{2.1}--\eqref{2.3}
for given charge and current density.\bigskip

We have now obtained the equations of field for given sources $\rho$ and $j$. On the other hand, in order to make a self-consistent
coupled system, we have to find the equations of motions of particles for given fields $D$ and $B$.
In the linear electromagnetic theory, it is well-known that particle trajectories are determined by the Lorentz equation.
Let $p$ be a momentum of a particle with charge $q$ and $v$ the corresponding velocity. Then, the Lorentz force exerted
on the particle is given as follows:
\[
\frac{dp}{dt}=q(E+v\times B),\quad\mbox{where}\quad v=\frac{p}{\sqrt{1+|p|^2}}.
\]
However, it is not easy to derive explicit equations of motions of particles
when they are interacting with the Born-Infeld electromagnetic field due to the nonlinearity.
We refer to \cite{Chr98,CK96,Ki94} for alternative methods.
In the case that a small electromagnetic field is considered, we can use the ``generalized Lorentz force'' as in
\cite{Che99,Che05}, where the author studied the dynamics of a particle ``dyon'' interacting with the Born-Infeld field.
Dyon is a hypothetical particle which has both electric and magnetic charges. Let $q_e$ be electric charge
and $q_m$ magnetic charge of a particle. Then we have the following generalized Lorentz force:
\[
\frac{dp}{dt}=q_e(D+v\times B)+q_m(B-v\times D).
\]
In the present paper, we will not consider the particles
having magnetic charges, hence we take the equation of motions of particles
as the following equation:
\begin{equation}\label{2.4}
\frac{dp}{dt}=q_e(D+v\times B).
\end{equation}
%
%
\subsection{One and one-half dimensional case}
We now set up the problem for the one and one-half dimensional case by following the framework of \cite{GSc90}, where
the authors studied the Vlasov-Maxwell system.
In this case, the spatial variable is $x=(x,0,0)\in\bbr^1$ and the momentum variables are $v=(v_1,v_2,0)\in\bbr^2$.
Hence, the electric field $D$ and the magnetic field $B$ are given by
\[
D=(D_1,D_2,0)\in\bbr^2\quad\mbox{and}\quad B=(0,0,B)\in\bbr^1,
\]
where $D_i,B:I\times\bbr^1\mapsto\bbr^1$, $i=1,2$, for a suitable time interval $I\subset[0,\infty)$,
and the charge density $\rho=\rho(t,x)$ and the current density $j=j(t,x)$ are given by
\[
\rho\in\bbr^1\quad\mbox{and}\quad j=(j_1,j_2,0)\in\bbr^2.
\]
Then, the field equations \eqref{2.1} are reduced to
\[
\begin{aligned}
&\partial_tD_1=-j_1,\\
&\partial_tD_2=-\partial_xH-j_2,\\
&\partial_xD_1=\rho,\\
&\partial_tB=-\partial_xE_2,
\end{aligned}
\]
and we note that $D_1$ is obtained from the first and the third equations:
\[
D_1(t,x)=-\int_0^tj_1(s,x)\,ds\quad\mbox{or}\quad
D_1(t,x)=\int_{-\infty}^x\rho(t,y)\,dy.
\]
These are consistent with each other because the densities $\rho$ and $j$ will be induced by the Vlasov equation, from which
the continuity equation ``$\partial_t\rho+\partial_xj_1=0$'' is naturally obtained, which makes the consistency between
them. We will use the latter formula to define $D_1$, and now we have only two equations of fields.
The constitutive relations \eqref{2.2}--\eqref{2.3} are reduced to
\[
\begin{aligned}
E&=\frac{(1+B^2)D}{\sqrt{1+|D|^2+B^2+|D|^2B^2}}=\sqrt{1+B^2}\frac{D}{\sqrt{1+|D|^2}},\\
H&=\frac{(1+|D|^2)B}{\sqrt{1+|D|^2+B^2+|D|^2B^2}}=\sqrt{1+|D|^2}\frac{B}{\sqrt{1+B^2}}.
\end{aligned}
\]
We now consider the Vlasov equation. In contrast with \cite{GSc90}, where the Lorentz force
was used, we take the generalized Lorentz force \eqref{2.4} to have the following equation:
\[
\partial_tf+\hat{v}_1\partial_xf+(D_1+\hat{v}_2B,D_2-\hat{v}_1B)\cdot\nabla_vf=0,
\]
where $v$ denotes the momentum of a particle, and $\hat{v}$ is the corresponding velocity: $\hat{v}=v/(1+|v|^2)^{1/2}$.
The charge density $\rho$ and the current density $j$ are induced by $f$ as follows:
\[
\rho(t,x)=\int_{\bbr^2}f(t,x,v)\,dv-n(x)\quad\mbox{and}\quad
j(t,x)=\int_{\bbr^2}\hat{v}f(t,x,v)\,dv,
\]
where $n(x)$ is a neutralizing background density in the sense that
\[
\int_{\bbr}\int_{\bbr^2}f^{in}(x,v)\,dv\,dx=\int_{\bbr}n(x)\,dx,
\]
where $f^{in}$ is an initial data of $f$, and the above quantity is preserved in time by the mass conservation of the Vlasov
equation.\bigskip

As a result, we will study the Cauchy problem of the following
system of PDEs. For $t\geq 0$, $x\in\bbr^1$, and $v\in\bbr^2$,
consider the distribution function $f=f(t,x,v)$, the electric field $D=D(t,x)\in\bbr^2$,
and the magnetic field $B=B(t,x)\in\bbr^1$ satisfying
\begin{align}
\partial_tf+\hat{v}_1\partial_xf+(D_1+\hat{v}_2B,D_2-\hat{v}_1B)\cdot\nabla_vf&=0,\label{2.5}\\
\partial_tD_2+\partial_x\left[\frac{(1+|D|^2)B}{\sqrt{1+|D|^2+B^2+|D|^2B^2}}\right]&=-j_2,\label{2.6}\\
\partial_tB+\partial_x\left[\frac{(1+B^2)D_2}{\sqrt{1+|D|^2+B^2+|D|^2B^2}}\right]&=0,\label{2.7}
\end{align}
where the charge density $\rho=\rho(t,x)$ and the current density $j=j(t,x)$ are given by
\begin{equation}\label{2.8}
\rho(t,x)=\int_{\bbr^2}f(t,x,v)\,dv-n(x)\quad\mbox{and}\quad
j(t,x)=\int_{\bbr^2}\hat{v}f(t,x,v)\,dv,
\end{equation}
and the first component of $D$ field is defined by $\rho$ as follows:
\begin{equation}\label{2.9}
D_1(t,x)=\int_{-\infty}^x\rho(t,y)\,dy.
\end{equation}
\begin{remark}
The system \eqref{2.5}--\eqref{2.9} can be thought of as a nonlinear version of the one and one-half dimensional Vlasov-Maxwell
system \cite{GSc90}. Maxwell's equations have been replaced by a system of
hyperbolic equations \eqref{2.6}--\eqref{2.7}.
\end{remark}
%
%
\section{Transformed system}\setcounter{equation}{0}
In this section, we study the field equations \eqref{2.6}--\eqref{2.7} for given charge density
$\rho$ and current density $j$. We introduce a useful transformation by which we can analyze \eqref{2.6}--\eqref{2.7}
more effectively. Note that \eqref{2.6}--\eqref{2.7} is a hyperbolic system with two unknowns $D_2$ and $B$
for given $\rho$ and $j$ since $\rho$ induces the first component of $D$ by \eqref{2.9}.
%
%
\subsection{Transformation of the field equations}
As a first step, we rewrite the hyperbolic system
\eqref{2.6}--\eqref{2.7} into a quasilinear form.
%
%
\begin{lemma}
Let $D_2=D_2(t,x)$ and $B=B(t,x)$ be $\mathcal{C}^1$-solutions to the hyperbolic system \eqref{2.6}--\eqref{2.7}
for given $\mathcal{C}^1$-functions $\rho$ and $j$
with $D_1$ defined by \eqref{2.9}.
Then, the hyperbolic system is rewritten as the following quasilinear form:
\begin{equation}\label{3.1}
\partial_t\left(
\begin{array}{c}
D_2 \\ B
\end{array}
\right)
+\left(
\begin{array}{cc}
A_{11} & A_{12} \\ A_{21} & A_{22}
\end{array}
\right)
\partial_x\left(
\begin{array}{c}
D_2 \\ B
\end{array}
\right)
=\left(
\begin{array}{c}
C_1\rho-j_2 \\ C_2\rho
\end{array}
\right),
\end{equation}
where the components of $A=(A_{ij})_{i,j=1,2}$ are given by
\[
\begin{aligned}
&A_{11}=\frac{D_2B}{\sqrt{1+|D|^2}\sqrt{1+B^2}},\quad
A_{12}=\sqrt{1+|D|^2}\frac{1}{(1+B^2)^{3/2}},\\
&A_{21}=\sqrt{1+B^2}\frac{(1+D_1^2)}{(1+|D|^2)^{3/2}},\quad
A_{22}=\frac{D_2B}{\sqrt{1+|D|^2}\sqrt{1+B^2}},
\end{aligned}
\]
and $C_i$, $i=1,2$, are given by
\[
C_1=-\frac{D_1B}{\sqrt{1+|D|^2}\sqrt{1+B^2}}\quad\mbox{and}\quad
C_2=\sqrt{1+B^2}\frac{D_1D_2}{(1+|D|^2)^{3/2}}.
\]
\end{lemma}
\begin{proof}
The proof is an elementary calculation, so we only remark the fact that $x$-derivatives of $D_1$ have been replaced by
$\rho$. Hence, we obtain a quasilinear form with respect to $D_2$ and $B$ for given $\rho$ and $j$.
\end{proof}
We next compute the eigenvectors and the eigenvalues
of the matrix $A$.
%
%
\begin{lemma}
Consider the matrix $A$ in \eqref{3.1}:
\[
A=\left(
\begin{array}{cc}
\displaystyle\frac{D_2B}{\sqrt{1+|D|^2}\sqrt{1+B^2}} &
\displaystyle\sqrt{1+|D|^2}\frac{1}{(1+B^2)^{3/2}}\\
\displaystyle\sqrt{1+B^2}\frac{(1+D_1^2)}{(1+|D|^2)^{3/2}}&
\displaystyle\frac{D_2B}{\sqrt{1+|D|^2}\sqrt{1+B^2}}
\end{array}
\right).
\]
Then, its left eigenvectors can be chosen by
\[
l_1=\left(\frac{\sqrt{1+D_1^2}}{1+|D|^2},\frac{-1}{1+B^2}\right)\quad\mbox{and}\quad
l_2=\left(\frac{\sqrt{1+D_1^2}}{1+|D|^2},\frac{1}{1+B^2}\right),
\]
and the corresponding eigenvalues are given by
\[
\lambda_1=\frac{D_2B-\sqrt{1+D_1^2}}{\sqrt{1+|D|^2}\sqrt{1+B^2}}\quad\mbox{and}\quad
\lambda_2=\frac{D_2B+\sqrt{1+D_1^2}}{\sqrt{1+|D|^2}\sqrt{1+B^2}}.
\]
\end{lemma}
\begin{proof}
This lemma is proved by direct calculations, so we skip the proof.
\end{proof}
%
%
\begin{remark}
If the fields $D$ and $B$ do not blow up at finite time, then the eigenvectors $l_1$ and $l_2$ are linearly independent, and
$\lambda_2$ is always strictly greater than $\lambda_1$. Thus,
we can see that the system \eqref{3.1} is strictly hyperbolic as long as its solution exists.
\end{remark}
Note that \eqref{3.1} is a quasilinear hyperbolic system with two unknown functions $D_2$ and $B$,
and this system is transformed into a diagonal form by the following transformation $\Phi_1$ and $\Phi_2$:
\begin{align*}
\Phi_1:\bbr^2\longrightarrow \left(-\frac{\pi}{2},\frac{\pi}{2}\right)\times\left(-\frac{\pi}{2},\frac{\pi}{2}\right)
\quad\mbox{and}\quad
\Phi_2:\bbr\longrightarrow\left(-\frac{\pi}{2},\frac{\pi}{2}\right)
\end{align*}
such that they are defined as follows: for $x=(x_1,x_2)\in\bbr^2$ and $y\in\bbr^1$, we define
\begin{align*}
\Phi_{1i}(x_1,x_2):=\arcsin\frac{x_i}{\sqrt{1+x_1^2+x_2^2}}\quad\mbox{and}\quad
\Phi_2(y):=\arcsin\frac{y}{\sqrt{1+y^2}}
\end{align*}
for $i=1,2$, where $\Phi_1=(\Phi_{11},\Phi_{12})$. Note that $\Phi_1$ and $\Phi_2$ are smooth
mappings and have smooth inverses on their images, $\Phi_1(\bbr^2)=\{(\phi_1,\phi_2):|\phi_1|+|\phi_2|<\pi/2\}$
and $\Phi_2(\bbr)=(-\pi/2,\pi/2)$ respectively.
We now define the transformed variables as follows:
\[
\theta_1:=\Phi_{11}(D_1,D_2),\quad
\theta_2:=\Phi_{12}(D_1,D_2),\quad\mbox{and}\quad
\theta_B:=\Phi_2(B).
\]
In other word, the following relations hold:
\begin{equation}\label{3.2}
\sin\theta_1=\frac{D_1}{\sqrt{1+|D|^2}},\quad\sin\theta_2=\frac{D_2}{\sqrt{1+|D|^2}},\quad
\sin\theta_B=\frac{B}{\sqrt{1+B^2}},
\end{equation}
which are illustrated in Figure \ref{f1}.

\begin{figure}[h]
\centering
\includegraphics[width=11.2cm]{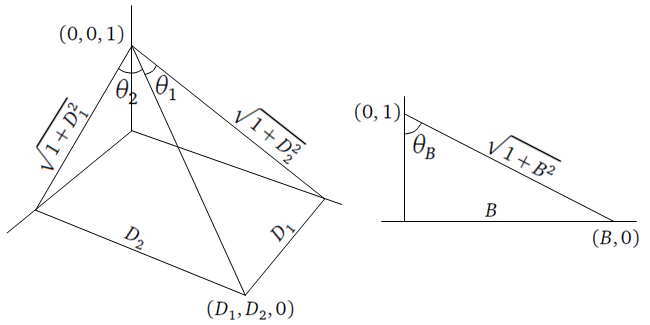}
\caption{The unknowns $-\pi/2<\theta_i<\pi/2$, $i\in\{1,2,B\}$, are well-defined
for finite $D_1$, $D_2$, and $B$. Note that if $|\theta_1|+|\theta_2|\rightarrow\pi/2$,
then either $D_1$, $D_2$, or both of them tend to infinity, which is also the case for $\theta_B$.
Hence, we can see that $\theta_i\neq\pm\pi/2$ unless the fields blow up.}
\label{f1}
\end{figure}

On the other hand, since $D_1$ is a given function of $\rho$,
we get the inverses as follows:
\[
D_2=\sqrt{1+D_1^2}\tan\theta_2\quad\mbox{and}\quad
B=\tan\theta_B.
\]
Using the above transformation, we will rewrite \eqref{3.1} as a hyperbolic system with two unknown functions
$\theta_2$ and $\theta_B$. To do that, $\theta_1$ is regarded as a function of $\theta_2$ and $D_1$, i.e.,
the following identities will be used in the proof of Lemma 3.3:
\begin{equation}\label{3.3}
\sin\theta_1=\frac{D_1}{\sqrt{1+D_1^2}}\cos\theta_2\quad\mbox{and}\quad
\frac{1}{\sqrt{1+|D|^2}}=\frac{1}{\sqrt{1+D_1^2}}\cos\theta_2.
\end{equation}
By the following lemma, the quasilinear hyperbolic system \eqref{3.1} is reduced to a diagonal form.
%
%
\begin{lemma}
Let $D_2$ and $B$ be $\mathcal{C}^1$-solutions to the quasilinear hyperbolic system \eqref{3.1},
and suppose that $\rho$ and $j$ are given $\mathcal{C}^1$-functions satisfying
\begin{equation}\label{3.4}
\partial_tD_1=-j_1\quad\mbox{and}\quad\partial_xD_1=\rho
\end{equation}
for $D_1$ defined by \eqref{2.9}.
Consider the transformation \eqref{3.2}. Then, the quasilinear system \eqref{3.1} is transformed to the following
inhomogeneous linearly degenerate system of a diagonal form for $\alpha=\alpha(t,x)$ and $\beta=\beta(t,x)$:
\begin{equation}\label{3.5}
\partial_t\left(
\begin{array}{c}
\alpha \\ \beta
\end{array}\right)
+\left(\begin{array}{cc}
-\cos\beta & 0 \\ 0 & \cos\alpha
\end{array}\right)
\partial_x\left(
\begin{array}{c}
\alpha \\ \beta
\end{array}\right)
=
\left(
\begin{array}{c}
\cos\theta_2(k_0\sin\theta_B+k_1\sin\theta_2+k_2\cos\theta_2) \\
\cos\theta_2(k_0\sin\theta_B+k_1\sin\theta_2+k_2\cos\theta_2)
\end{array}\right),
\end{equation}
where $\alpha=\theta_2-\theta_B$ and $\beta=\theta_2+\theta_B$, and $k_i$, $i=0,1,2$, are given by
\[
k_0=-\frac{\rho D_1}{1+D_1^2},\quad k_1=\frac{j_1D_1}{1+D_1^2},\quad k_2=-\frac{j_2}{\sqrt{1+D_1^2}}.
\]
\end{lemma}
\begin{proof}
We first transform the eigenvalues $\lambda_i$, $i=1,2$, as follows:
\[
\begin{aligned}
&\lambda_1=\sin\theta_2\sin\theta_B-\cos\theta_2\cos\theta_B=-\cos(\theta_2+\theta_B)=-\cos\beta,\\
&\lambda_2=\sin\theta_2\sin\theta_B+\cos\theta_2\cos\theta_B=\cos(\theta_2-\theta_B)=\cos\alpha.
\end{aligned}
\]
We now consider the eigenvectors $l_i$, $i=1,2$. From the well-known transformation of polar coordinates,
$(x,y)\in\bbr^2\leftrightarrow(r,\theta)\in [0,\infty)\times[0,2\pi)$,
where $x=r\cos\theta$ and $y=r\sin\theta$, we know that
\[
\partial_x\theta=-\frac{1}{r}\sin\theta\quad\mbox{and}\quad
\partial_y\theta=\frac{1}{r}\cos\theta.
\]
Using the above relations, we first calculate the following quantities.
As we can see from Figure \ref{f1}, $\theta_2$ can be thought of a function of $\sqrt{1+D_1^2}$ and $D_2$.
\begin{equation}\label{3.6}
\begin{aligned}
\frac{\partial\theta_2}{\partial D_1}&=\frac{\partial\theta_2}{\partial(\sqrt{1+D_1^2})}\frac{\partial(\sqrt{1+D_1^2})}{\partial D_1}\\
&=-\frac{1}{\sqrt{1+|D|^2}}\sin\theta_2\frac{D_1}{\sqrt{1+D_1^2}}
=-\frac{D_1}{1+D_1^2}\cos\theta_2\sin\theta_2,
\end{aligned}
\end{equation}
where we used \eqref{3.3}.
\begin{equation}\label{3.7}
\frac{\partial\theta_2}{\partial D_2}=\frac{1}{\sqrt{1+|D|^2}}\cos\theta_2=\frac{\sqrt{1+D_1^2}}{1+|D|^2},
\end{equation}
\begin{equation}\label{3.8}
\frac{\partial\theta_B}{\partial B}=\frac{1}{\sqrt{1+B^2}}\cos\theta_B=\frac{1}{1+B^2}.
\end{equation}
Hence, the eigenvectors $l_i$, $i=1,2$, are transformed to
\[
l_1=\left(\frac{\partial\theta_2}{\partial D_2},-\frac{\partial\theta_B}{\partial B}\right)\quad\mbox{and}\quad
l_2=\left(\frac{\partial\theta_2}{\partial D_2},\frac{\partial\theta_B}{\partial B}\right).
\]
We now multiply the transformed eigenvector $l_1$ to \eqref{3.1} and obtain
\begin{eqnarray*}
\lefteqn{\left(\frac{\partial\theta_2}{\partial D_2}\frac{\partial D_2}{\partial t}
-\frac{\partial\theta_B}{\partial B}\frac{\partial B}{\partial t}\right)
+\lambda_1\left(\frac{\partial\theta_2}{\partial D_2}\frac{\partial D_2}{\partial x}
-\frac{\partial\theta_B}{\partial B}\frac{\partial B}{\partial x}\right)}\cr
&=&-\frac{\partial\theta_2}{\partial D_2}\frac{D_1B}{\sqrt{1+|D|^2}\sqrt{1+B^2}}\rho
-\frac{\partial\theta_2}{\partial D_2}j_2
-\frac{\partial\theta_B}{\partial B}\sqrt{1+B^2}\frac{D_1D_2}{(1+|D|^2)^{3/2}}\rho,
\end{eqnarray*}
and then we add the following quantity on both sides:
\[
\frac{\partial\theta_2}{\partial D_1}\frac{\partial D_1}{\partial t}
+\lambda_1\frac{\partial\theta_2}{\partial D_1}\frac{\partial D_1}{\partial x}
\]
to get the desired result by \eqref{3.6}, \eqref{3.7} and \eqref{3.8}.
\begin{eqnarray*}
\lefteqn{\left(\frac{\partial\theta_2}{\partial t}-\frac{\partial\theta_B}{\partial t}\right)
+\lambda_1\left(\frac{\partial\theta_2}{\partial x}-\frac{\partial\theta_B}{\partial x}\right)}\cr
&=&-\frac{\partial\theta_2}{\partial D_2}\frac{D_1B}{\sqrt{1+|D|^2}\sqrt{1+B^2}}\rho
-\frac{\partial\theta_2}{\partial D_2}j_2
-\frac{\partial\theta_B}{\partial B}\sqrt{1+B^2}\frac{D_1D_2}{(1+|D|^2)^{3/2}}\rho\cr
&&{}+\frac{\partial\theta_2}{\partial D_1}\frac{\partial D_1}{\partial t}
+\lambda_1\frac{\partial\theta_2}{\partial D_1}\frac{\partial D_1}{\partial x}\cr
&=&-\frac{\sqrt{1+D_1^2}}{1+|D|^2}\frac{D_1B}{\sqrt{1+|D|^2}\sqrt{1+B^2}}\rho
-\frac{\sqrt{1+D_1^2}}{1+|D|^2}j_2
-\frac{1}{1+B^2}\sqrt{1+B^2}\frac{D_1D_2}{(1+|D|^2)^{3/2}}\rho\cr
&&{}+\frac{D_1}{1+D_1^2}\cos\theta_2\sin\theta_2j_1+\cos(\theta_2+\theta_B)\frac{D_1}{1+D_1^2}\cos\theta_2\sin\theta_2\rho,
\end{eqnarray*}
where we used \eqref{3.4} for the last two quantities.
We now use \eqref{3.3} together with the transformation \eqref{3.2} to get
\begin{eqnarray*}
\lefteqn{\left(\frac{\partial\theta_2}{\partial t}-\frac{\partial\theta_B}{\partial t}\right)
+\lambda_1\left(\frac{\partial\theta_2}{\partial x}-\frac{\partial\theta_B}{\partial x}\right)}\cr
&=&-\frac{1}{\sqrt{1+D_1^2}}\cos^2\theta_2\sin\theta_1\sin\theta_B\rho
-\frac{1}{\sqrt{1+D_1^2}}\cos^2\theta_2j_2
-\cos\theta_B\sin\theta_1\sin\theta_2\frac{1}{\sqrt{1+|D|^2}}\rho\cr
&&{}+\frac{D_1}{1+D_1^2}\cos\theta_2\sin\theta_2j_1+\cos(\theta_2+\theta_B)\frac{D_1}{1+D_1^2}\cos\theta_2\sin\theta_2\rho\cr
&=&-\frac{D_1}{1+D_1^2}\cos^3\theta_2\sin\theta_B\rho
-\frac{1}{\sqrt{1+D_1^2}}\cos^2\theta_2j_2
-\cos\theta_B\frac{D_1}{\sqrt{1+D_1^2}}\cos\theta_2\sin\theta_2\frac{1}{\sqrt{1+D_1^2}}\cos\theta_2\rho\cr
&&{}+\frac{D_1}{1+D_1^2}\cos\theta_2\sin\theta_2j_1+\cos(\theta_2+\theta_B)\frac{D_1}{1+D_1^2}\cos\theta_2\sin\theta_2\rho\cr
&=&-\frac{\rho D_1}{1+D_1^2}\cos^2\theta_2(\cos\theta_2\sin\theta_B+\sin\theta_2\cos\theta_B)
-\frac{j_2}{\sqrt{1+D_1^2}}\cos^2\theta_2\cr
&&{}+\frac{j_1D_1}{1+D_1^2}\cos\theta_2\sin\theta_2+\cos(\theta_2+\theta_B)\frac{\rho D_1}{1+D_1^2}\cos\theta_2\sin\theta_2\cr
&=&\frac{j_1D_1}{1+D_1^2}\cos\theta_2\sin\theta_2
-\frac{j_2}{\sqrt{1+D_1^2}}\cos^2\theta_2\cr
&&{}-\frac{\rho D_1}{1+D_1^2}\cos^2\theta_2\sin(\theta_2+\theta_B)
+\cos(\theta_2+\theta_B)\frac{\rho D_1}{1+D_1^2}\cos\theta_2\sin\theta_2\cr
&=&\frac{j_1D_1}{1+D_1^2}\cos\theta_2\sin\theta_2
-\frac{j_2}{\sqrt{1+D_1^2}}\cos^2\theta_2\cr
&&{}-\frac{\rho D_1}{1+D_1^2}\cos\theta_2
\Big(\cos\theta_2\sin(\theta_2+\theta_B)
-\cos(\theta_2+\theta_B)\sin\theta_2\Big)\cr
&=&\frac{j_1D_1}{1+D_1^2}\cos\theta_2\sin\theta_2
-\frac{j_2}{\sqrt{1+D_1^2}}\cos^2\theta_2
-\frac{\rho D_1}{1+D_1^2}\cos\theta_2\sin\theta_B,
\end{eqnarray*}
where we used trigonometric identities several times, and the first equation of \eqref{3.5}
is obtained.
\[
\partial_t\alpha-(\cos\beta)\partial_x\alpha
=\cos\theta_2\left(\frac{j_1D_1}{1+D_1^2}\sin\theta_2
-\frac{j_2}{\sqrt{1+D_1^2}}\cos\theta_2
-\frac{\rho D_1}{1+D_1^2}\sin\theta_B\right).
\]
The second equation of \eqref{3.5} is obtained by the same calculation.
Multiply the transformed eigenvector $l_2$ to \eqref{3.1} to get
\begin{eqnarray*}
\lefteqn{\left(\frac{\partial\theta_2}{\partial D_2}\frac{\partial D_2}{\partial t}
+\frac{\partial\theta_B}{\partial B}\frac{\partial B}{\partial t}\right)
+\lambda_2\left(\frac{\partial\theta_2}{\partial D_2}\frac{\partial D_2}{\partial x}
+\frac{\partial\theta_B}{\partial B}\frac{\partial B}{\partial x}\right)}\cr
&=&-\frac{\partial\theta_2}{\partial D_2}\frac{D_1B}{\sqrt{1+|D|^2}\sqrt{1+B^2}}\rho
-\frac{\partial\theta_2}{\partial D_2}j_2
+\frac{\partial\theta_B}{\partial B}\sqrt{1+B^2}\frac{D_1D_2}{(1+|D|^2)^{3/2}}\rho,
\end{eqnarray*}
and add the following quantity on both sides:
\[
\frac{\partial\theta_2}{\partial D_1}\frac{\partial D_1}{\partial t}
+\lambda_2\frac{\partial\theta_2}{\partial D_1}\frac{\partial D_1}{\partial x}
=-\frac{\partial\theta_2}{\partial D_1}j_1+\lambda_2\frac{\partial\theta_2}{\partial D_1}\rho
\]
to get
\begin{eqnarray*}
\lefteqn{\left(\frac{\partial\theta_2}{\partial t}+\frac{\partial\theta_B}{\partial t}\right)
+\lambda_2\left(\frac{\partial\theta_2}{\partial x}+\frac{\partial\theta_B}{\partial x}\right)}\cr
&=&-\frac{\partial\theta_2}{\partial D_2}\frac{D_1B}{\sqrt{1+|D|^2}\sqrt{1+B^2}}\rho
-\frac{\partial\theta_2}{\partial D_2}j_2
+\frac{\partial\theta_B}{\partial B}\sqrt{1+B^2}\frac{D_1D_2}{(1+|D|^2)^{3/2}}\rho\cr
&&{}-\frac{\partial\theta_2}{\partial D_1}j_1
+\lambda_2\frac{\partial\theta_2}{\partial D_1}\rho\cr
&=&-\frac{\sqrt{1+D_1^2}}{1+|D|^2}\frac{D_1B}{\sqrt{1+|D|^2}\sqrt{1+B^2}}\rho
-\frac{\sqrt{1+D_1^2}}{1+|D|^2}j_2
+\frac{1}{1+B^2}\sqrt{1+B^2}\frac{D_1D_2}{(1+|D|^2)^{3/2}}\rho\cr
&&{}+\frac{D_1}{1+D_1^2}\cos\theta_2\sin\theta_2j_1
-\cos(\theta_2-\theta_B)\frac{D_1}{1+D_1^2}\cos\theta_2\sin\theta_2\rho\cr
&=&-\frac{1}{\sqrt{1+D_1^2}}\cos^2\theta_2\sin\theta_1\sin\theta_B\rho
-\frac{1}{\sqrt{1+D_1^2}}\cos^2\theta_2j_2
+\cos\theta_B\sin\theta_1\sin\theta_2\frac{1}{\sqrt{1+|D|^2}}\rho\cr
&&{}+\frac{D_1}{1+D_1^2}\cos\theta_2\sin\theta_2j_1-\cos(\theta_2-\theta_B)\frac{D_1}{1+D_1^2}\cos\theta_2\sin\theta_2\rho\cr
&=&-\frac{D_1}{1+D_1^2}\cos^3\theta_2\sin\theta_B\rho
-\frac{1}{\sqrt{1+D_1^2}}\cos^2\theta_2j_2
+\cos\theta_B\frac{D_1}{\sqrt{1+D_1^2}}\cos\theta_2\sin\theta_2\frac{1}{\sqrt{1+D_1^2}}\cos\theta_2\rho\cr
&&{}+\frac{D_1}{1+D_1^2}\cos\theta_2\sin\theta_2j_1-\cos(\theta_2-\theta_B)\frac{D_1}{1+D_1^2}\cos\theta_2\sin\theta_2\rho\cr
&=&-\frac{\rho D_1}{1+D_1^2}\cos^2\theta_2(\cos\theta_2\sin\theta_B-\sin\theta_2\cos\theta_B)
-\frac{j_2}{\sqrt{1+D_1^2}}\cos^2\theta_2\cr
&&{}+\frac{j_1D_1}{1+D_1^2}\cos\theta_2\sin\theta_2-\cos(\theta_2-\theta_B)\frac{\rho D_1}{1+D_1^2}\cos\theta_2\sin\theta_2\cr
&=&\frac{j_1D_1}{1+D_1^2}\cos\theta_2\sin\theta_2
-\frac{j_2}{\sqrt{1+D_1^2}}\cos^2\theta_2\cr
&&{}+\frac{\rho D_1}{1+D_1^2}\cos^2\theta_2\sin(\theta_2-\theta_B)
-\cos(\theta_2-\theta_B)\frac{\rho D_1}{1+D_1^2}\cos\theta_2\sin\theta_2\cr
&=&\frac{j_1D_1}{1+D_1^2}\cos\theta_2\sin\theta_2
-\frac{j_2}{\sqrt{1+D_1^2}}\cos^2\theta_2\cr
&&{}+\frac{\rho D_1}{1+D_1^2}\cos\theta_2
\Big(\cos\theta_2\sin(\theta_2-\theta_B)
-\cos(\theta_2-\theta_B)\sin\theta_2\Big)\cr
&=&\frac{j_1D_1}{1+D_1^2}\cos\theta_2\sin\theta_2
-\frac{j_2}{\sqrt{1+D_1^2}}\cos^2\theta_2
-\frac{\rho D_1}{1+D_1^2}\cos\theta_2\sin\theta_B.
\end{eqnarray*}
Hence, we obtain the second equation of \eqref{3.5} as follows:
\[
\partial_t\beta+(\cos\alpha)\partial_x\beta
=\cos\theta_2\left(\frac{j_1D_1}{1+D_1^2}\sin\theta_2
-\frac{j_2}{\sqrt{1+D_1^2}}\cos\theta_2
-\frac{\rho D_1}{1+D_1^2}\sin\theta_B\right).
\]
Note that the inhomogeneous term is the same with the first equation of \eqref{3.5},
and this completes the proof.
\end{proof}
To summarize, the hyperbolic system \eqref{2.6}--\eqref{2.7} was rewritten as a quasilinear form \eqref{3.1},
and we transformed it to a diagonal form \eqref{3.5} by using the transformation \eqref{3.2}.
%
%
\subsection{Quasilinear hyperbolic systems}
Quasilinear hyperbolic systems, especially those which can be written as hyperbolic conservation laws,
have been extensively studied for last several decades. We refer to some good textbooks \cite{Br,Le,Se,Sm}
and the references contained therein. In this paper, we will use a classical result of Hartman-Wintner \cite{HW52}
where local-in-time existence of $\mathcal{C}^1$ solutions to inhomogeneous quasilinear hyperbolic systems was proved.
%
%
\begin{proposition}\cite{Do52,HW52}
Let $t\geq 0$, $x\in\bbr$, and $u=u(t,x)\in\bbr^n$.
Consider the Cauchy problem for the following inhomogeneous quasilinear hyperbolic system
with $\mathcal{C}^1$ initial data $u^{in}$:
\begin{equation}\label{3.9}
u_t=Fu_x+G,\quad
u(0,x)=u^{in}(x)
\end{equation}
on the following $(t,x,u)$-region:
\begin{equation}\label{3.10}
0\leq t\leq t',\quad
|x|\leq b,
\end{equation}
\begin{equation}\label{3.11}
|u_1|\leq c,\cdots,|u_n|\leq c.
\end{equation}
Suppose that the system \eqref{3.9} is strictly hyperbolic on \eqref{3.10}--\eqref{3.11}, where
$F=F(t,x,u)$ is a $\mathcal{C}^1$ matrix function and $G(t,x,u)$ is a continuous vector function
on \eqref{3.10}--\eqref{3.11},
and $G$ is of class $\mathcal{C}^1$ with respect to $x$, $u_1,\cdots,u_n$.
Then there exists a positive $t_*\leq t'$ such that the system \eqref{3.9}
has a unique $\mathcal{C}^1$ solution on the region
\begin{equation}\label{3.12}
0\leq t\leq t_*\quad\mbox{and}\quad |x|+Kt\leq b,
\end{equation}
where $K$ is chosen such that $|\lambda_i|\leq K$, $i=1,\cdots,n$, on \eqref{3.10}.
\end{proposition}
\begin{proof}
We refer to \cite{HW52} for the proof.
\end{proof}
By applying Proposition 3.1 to our case \eqref{3.5}, we obtain the following corollary.
A lower bound $t_*$ can be determined by following the proof of \cite{HW52};
moreover, determining it is much easier than that of \cite{HW52}
because the $2\times 2$ matrix in \eqref{3.5} is already diagonal.
%
%
\begin{corollary}
Consider the following inhomogeneous quasilinear hyperbolic system \eqref{3.5}:
\[
\partial_t\left(
\begin{array}{c}
\alpha \\ \beta
\end{array}\right)
+\left(\begin{array}{cc}
-\cos\beta & 0 \\ 0 & \cos\alpha
\end{array}\right)
\partial_x\left(
\begin{array}{c}
\alpha \\ \beta
\end{array}\right)
=
\left(
\begin{array}{c}
\cos\theta_2(k_0\sin\theta_B+k_1\sin\theta_2+k_2\cos\theta_2) \\
\cos\theta_2(k_0\sin\theta_B+k_1\sin\theta_2+k_2\cos\theta_2)
\end{array}\right)
\]
on the following region:
\begin{equation}\label{3.13}
|\theta_2(t,x)|+|\theta_B(t,x)|<\frac{\pi}{2},
\end{equation}
where $\alpha=\theta_2-\theta_B$ and $\beta=\theta_2+\theta_B$.
Suppose that $\alpha^{in}$ and $\beta^{in}$ are of class $\mathcal{C}^1$ satisfying \eqref{3.13} at $t=0$ and
compactly supported in $\{x\in\bbr: |x|\leq P\}$.
Let $k_i=k_i(t,x)$, $i=0,1,2$, be continuous and of class $\mathcal{C}^1$ with respect to $x$ variable,
and assume that there exists a positive nondecreasing function $\tilde{k}$ such that
\[
||k_0(t)||_\infty+||k_1(t)||_\infty+||k_2(t)||_\infty\leq \tilde{k}(t).
\]
Then there exists a unique $\mathcal{C}^1$ solution on a time interval $[0,t_*)$ where $t_*$ is defined as follows:
\[
t_*:=\sup\left\{t:
||\theta_2^{in}||_\infty+||\theta_B^{in}||_\infty
+\int_0^t\tilde{k}(\tau)\,d\tau<\frac{\pi}{2}\right\}.
\]
\end{corollary}
\begin{proof}
We note that the system \eqref{3.5} is strictly hyperbolic on the region \eqref{3.13}:
\[
\lambda_2-\lambda_1=\cos\alpha+\cos\beta=2\cos\theta_2\cos\theta_B>0.
\]
The other conditions in Proposition 3.1 hold, hence we have local-in-time existence of $\mathcal{C}^1$ solutions
by applying the proposition.

In contrast with homogeneous cases, the solution $u$ of \eqref{3.9} will grow
due to the inhomogeneous term as time evolves.
For this reason, the domain \eqref{3.10} should be
restricted as \eqref{3.12},
i.e., the solution $u$ satisfies \eqref{3.11} only on the restricted region \eqref{3.12},
and thus we have to find a time interval where the solutions $\alpha$ and $\beta$ satisfy \eqref{3.13}.
We use the integral formula \eqref{5.5}--\eqref{5.6} in Section 5 and estimate them as follows.
We first note that
\[
|\theta_2(t,x)|+|\theta_B(t,x)|<\frac{\pi}{2}\quad\Longleftrightarrow\quad
|\alpha(t,x)|<\frac{\pi}{2}\quad\mbox{and}\quad
|\beta(t,x)|<\frac{\pi}{2}.
\]
By the integral formula \eqref{5.5}, we have
\[
\begin{aligned}
|\alpha(t,x)|&\leq||\alpha^{in}||_\infty
+\int_0^t||k_0(\tau)||_\infty+||k_1(\tau)||_\infty+||k_2(\tau)||_\infty\,d\tau\\
&\leq||\theta_2^{in}||_\infty+||\theta_B^{in}||_\infty
+\int_0^t||k_0(\tau)||_\infty+||k_1(\tau)||_\infty+||k_2(\tau)||_\infty\,d\tau.
\end{aligned}
\]
Similarly, we have the same estimate for $\beta$.
\[
|\beta(t,x)|\leq||\theta_2^{in}||_\infty+||\theta_B^{in}||_\infty
+\int_0^t||k_0(\tau)||_\infty+||k_1(\tau)||_\infty+||k_2(\tau)||_\infty\,d\tau.
\]
Therefore, \eqref{3.13} holds on the time interval $[0,t_*)$ by the definition of $t_*$,
and the solutions exist on this time interval.
\end{proof}
%
%
\section{Main result}\setcounter{equation}{0}
In the previous section, we derived the system of equations of our interests.
The system of \eqref{2.5}--\eqref{2.9} is transformed by \eqref{3.2} as follows:
\begin{align}
&\partial_tf+\hat{v}_1\partial_xf+(D_1+\hat{v}_2B,D_2-\hat{v}_1B)\cdot\nabla_vf=0,\label{4.1}\\
&\partial_t\alpha-(\cos\beta)\partial_x\alpha=\cos\theta_2(k_0\sin\theta_B+k_1\sin\theta_2+k_2\cos\theta_2),\label{4.2}\\
&\partial_t\beta+(\cos\alpha)\partial_x\beta=\cos\theta_2(k_0\sin\theta_B+k_1\sin\theta_2+k_2\cos\theta_2),\label{4.3}
\end{align}
where $\alpha=\theta_2-\theta_B$ and $\beta=\theta_2+\theta_B$, and $k_i$, $i=0,1,2$, are given by
\begin{equation}\label{4.4}
k_0=-\frac{\rho D_1}{1+D_1^2},\quad
k_1=\frac{j_1D_1}{1+D_1^2},\quad
k_2=-\frac{j_2}{\sqrt{1+D_1^2}}.
\end{equation}
They are coupled to each other as follows:
\begin{equation}\label{4.5}
\begin{aligned}
&\rho(t,x)=\int_{\bbr^2}f(t,x,v)\,dv-n(x)\quad\mbox{and}\quad D_1(t,x)=\int_{-\infty}^x\rho(t,y)\,dy,\\
&j(t,x)=\int_{\bbr^2}\hat{v}f(t,x,v)\,dv,
\end{aligned}
\end{equation}
and the electromagnetic fields $D_2$ and $B$ are given by
\begin{equation}\label{4.6}
D_2=\sqrt{1+D_1^2}\tan\theta_2\quad\mbox{and}\quad
B=\tan\theta_B.
\end{equation}
Let $f^{in}$, $\theta_2^{in}$, and $\theta_B^{in}$ be initial data
for the unknowns $f$, $\theta_2$, and $\theta_B$, or equivalently
$f$, $\alpha$, and $\beta$, and we consider the Cauchy problem for \eqref{4.1}--\eqref{4.6}.\bigskip

\noindent{\bf Assumptions on initial data.}
To study the Cauchy problem for \eqref{4.1}--\eqref{4.6}, suitable assumptions on initial data should be taken.
We will assume that initial data $f^{in}$, $\theta_2^{in}$, and $\theta_B^{in}$ satisfy the following conditions A1--3:
\begin{itemize}
\item[A1.] They are compactly supported $\mathcal{C}^1$ functions:
\[
f^{in}\in\mathcal{C}^1_c(\bbr\times\bbr^2),\quad
\theta_2^{in}\in\mathcal{C}^1_c(\bbr),\quad
\theta_B^{in}\in\mathcal{C}^1_c(\bbr),\quad n\in\mathcal{C}^1_c(\bbr),
\]
and $P$ denotes the size of their support, i.e., they are supported in
\[
\left\{(x,v):|x|, |v|\leq P\right\}\quad \mbox{or}\quad \left\{x:|x|\leq P\right\}.
\]
\item[A2.] The fields $\theta_2^{in}$ and $\theta_B^{in}$ are bounded such that
\[
|\theta_2^{in}(x)|+ |\theta_B^{in}(x)|<\frac{\pi}{2}\quad\mbox{for}\quad |x|\leq P.
\]
\item[A3.] The fields are small in the sense that
\[
||\theta_2^{in}||_\infty+||\theta_B^{in}||_\infty<\arctan\frac{1}{P}.
\]
\end{itemize}
Note that there are two kinds of characteristic curves,
one from the Vlasov equation \eqref{4.1} and the other ones from the field equations \eqref{4.2}--\eqref{4.3}.
The particles described by the Vlasov equation have velocities $\hat{v}$, while the Born-Infeld field propagates
with speed $-\cos\beta$ or $\cos\alpha$, which may be less than the speed of light $c=1$.
Hence, there may be a resonance between a particle trajectory and field propagation, and
in this case we cannot apply the arguments of \cite{GSt86}:
for the Vlasov-Maxwell case, the Vlasov and the Maxwell characteristics are always linearly independent,
so $x$ or $t$ derivatives can be decomposed into linear combinations of the Vlasov and the Maxwell characteristics
(see \cite{Gl,GSt86} for details).
This argument can be applied to our case \eqref{4.1}--\eqref{4.3} when we take the third assumption above.
The assumption A3 implies that the Vlasov and the Born-Infeld characteristics are separated at $t=0$. Note that
\[
||\theta_2^{in}||_\infty+||\theta_B^{in}||_\infty<\arctan\frac{1}{P}\quad\Longleftrightarrow\quad
\hat{P}<\cos\left(||\theta_2^{in}||_\infty+||\theta_B^{in}||_\infty\right),
\]
and this implies that under the assumptions A1 and A2 we have
\[
|\hat{v}_1|<\cos\left(\max\left\{|\alpha^{in}(x)|,|\beta^{in}(x)|\right\}\right)
=\min\left\{\cos\alpha^{in}(x),\cos\beta^{in}(x)\right\}
\]
for any $x$ and $v$ satisfying $|x|\leq P$ and $|v|\leq P$.
Remind that $\alpha^{in}=\theta_2^{in}-\theta_B^{in}$ and $\beta^{in}=\theta_2^{in}+\theta_B^{in}$,
and the cosine function is decreasing on $[0,\pi/2]$.
Hence, the particle velocity and the field propagation speed are initially separated as follows:
\begin{equation}
-\cos\beta^{in}(x)<\hat{v}_1<\cos\alpha^{in}(x).
\end{equation}
We now state the main theorem of the present paper.
%
%
\begin{theorem}
Consider the Cauchy problem for the system \eqref{4.1}--\eqref{4.6} When initial data
$f^{in}$, $\theta_2^{in}$, and $\theta_B^{in}$ satisfy the assumptions A1--3, then the system
\eqref{4.1}--\eqref{4.6} has a unique $\mathcal{C}^1$ solutions $f$, $\theta_2$, and $\theta_B$ locally in time.
\end{theorem}
\begin{proof}
The uniqueness is proved by following standard arguments for the Vlasov equations,
and we only prove the existence part of this main theorem.
The following sections are devoted to the proof of existence.
We refer to Section 5 and 6 for the proof.
\end{proof}
%
%
\section{A priori estimates}\setcounter{equation}{0}
In this section, we obtain several a priori estimates for $f$, $\alpha$, $\beta$, and their derivatives.
The arguments basically follow the well-known results \cite{GSc90,GSt86}.
It turns out that
the nonlinearity of the field, i.e., the necessity of controlling the field characteristics,
complicates the problem and requires some restrictions, which will be seen in the proofs of lemmas.

To obtain a priori estimates for the system \eqref{4.1}--\eqref{4.6}, we consider a slightly modified version
of it in this section: for given $D^*$ and $B^*$, we consider the following Vlasov equation:
\begin{equation}\label{5.1}
\partial_tf+\hat{v}_1\partial_xf+(D_1^*+\hat{v}_2B^*,D_2^*-\hat{v}_1B^*)\cdot\nabla_vf=0
\end{equation}
together with \eqref{4.2}--\eqref{4.6}.
We first consider the characteristic equations of the Vlasov and the field equations.
Suppose that $D^*$ and $B^*$ fields are $\mathcal{C}^1$ on a time interval $I$.
Then, the following system of ODEs has a unique solution $X(s;t,x,v)$ and $V(s;t,x,v)$ which are $\mathcal{C}^1$
on $I\times I\times\bbr\times\bbr^2$:
\begin{equation}\label{5.2}
\begin{aligned}
&\frac{d}{ds}X(s)=\hat{V}_1(s),\quad X(t)=x,\\
&\frac{d}{ds}V(s)=D^*(s,X(s))+(\hat{V}_2(s),-\hat{V}_1(s))B^*(s,X(s)),\quad V(t)=v,
\end{aligned}
\end{equation}
where $X(s)=X(s;t,x,v)$ and $V(s)=V(s;t,x,v)$. The Vlasov equation \eqref{5.1} is then solved as follows:
\begin{equation}\label{5.3}
f(t,x,v)=f^{in}(X(0),V(0)).
\end{equation}
It is well known that the following map is measure preserving:
\[
(x,v)\mapsto (X(s;t,x,v),V(s;t,x,v)),
\]
which implies that $||f(t)||_p=||f^{in}||_p$ for any $1\leq p\leq \infty$ and $t\in I$,
and we refer to \cite{GSc90} for details.

Similarly, if $\alpha$ and $\beta$ are $\mathcal{C}^1$ on $I$,
then the following two ODEs have $\mathcal{C}^1$ solutions $\xi(\tau;t,x)$ and $\eta(\tau;t,x)$
on $I\times I\times\bbr$:
\begin{align}
&\frac{d}{d\tau}\xi(\tau)=-\cos\beta(\tau,\xi(\tau)),\quad \xi(t)=x,\label{5.4}\\
&\frac{d}{d\tau}\eta(\tau)=\cos\alpha(\tau,\eta(\tau)),\quad \eta(t)=x.\label{5.5}
\end{align}
Therefore, $\alpha$ and $\beta$ in \eqref{4.2}--\eqref{4.3} satisfy the following integral formulae on $I\times\bbr$:
\begin{equation}\label{5.6}
\begin{aligned}
\alpha(t,x)=\alpha^{in}(\xi(0))
+\int_0^t\Big(&k_0(\tau,\xi(\tau))\sin\theta_B(\tau,\xi(\tau))
+k_1(\tau,\xi(\tau))\sin\theta_2(\tau,\xi(\tau))\\ &{}+k_2(\tau,\xi(\tau))\cos\theta_2(\tau,\xi(\tau))\Big)\cos\theta_2(\tau,\xi(\tau))\,d\tau,
\end{aligned}
\end{equation}
\begin{equation}\label{5.7}
\begin{aligned}
\beta(t,x)=\beta^{in}(\eta(0))
+\int_0^t\Big(&k_0(\tau,\eta(\tau))\sin\theta_B(\tau,\eta(\tau))
+k_1(\tau,\eta(\tau))\sin\theta_2(\tau,\eta(\tau))\\ &{}+k_2(\tau,\eta(\tau))\cos\theta_2(\tau,\eta(\tau))\Big)\cos\theta_2(\tau,\eta(\tau))\,d\tau.
\end{aligned}
\end{equation}
%
%
\subsection{A priori estimates for the unknowns and momentum support}
We first define the momentum support $P(t)$ of $f$ as follows:
\[
P(t):=\sup\left\{|v|:f(s,x,v)\neq 0,\ 0\leq s\leq t,\ x\in\bbr,\ v\in\bbr^2\right\}.
\]
Note that $P(0)=P$.
For $\mathcal{C}^1$ functions $D^*$ and $B^*$ which are defined on a time interval $I$,
the momentum support $P(t)$ is well defined on $I$ because the characteristic equation \eqref{5.2}
gives
\begin{equation}\label{5.8}
\sup_{f^{in}(x,v)\neq 0}|V(t;0,x,v)|\leq P +\int_0^t ||D^*(s)||_\infty+||B^*(s)||_\infty\,ds,\quad t\in I.
\end{equation}
The time interval $I$ will be determined later, and we present some useful lemmas.
The following two lemmas are easily obtained by using \eqref{4.4}, \eqref{4.5}, and \eqref{5.3}.
%
%
\begin{lemma}
Suppose that $D^*$ and $B^*$ are $\mathcal{C}^1$ functions on a time interval $I$,
and consider the equations \eqref{5.1} and \eqref{4.5}.
Then $f$, $\rho$, $j$, and $D_1$ are well defined and $\mathcal{C}^1$ on $I$ and satisfy
the following estimates on $I\times\bbr\times\bbr^2$ or $I\times\bbr$:
\begin{align*}
(i)&\quad f(t,x,v)\leq ||f^{in}||_\infty.\cr
(ii)&\quad |\rho(t,x)|\leq \pi||f^{in}||_\infty P^2(t) +||n||_\infty.\cr
(iii)&\quad |j(t,x)|\leq \pi||f^{in}||_\infty P^2(t).\cr
(iv)&\quad |D_1(t,x)|\leq ||f^{in}||_1+||n||_1.
\end{align*}
\end{lemma}
\begin{proof}
We obtain the first estimate by using \eqref{5.3}. The second estimate is obtained by the definition of
$\rho$.
\begin{align*}
|\rho(t,x)|&\leq\int_{\bbr^2}f(t,x,v)\,dv +|n(x)|\\
&\leq
||f^{in}||_\infty\int_{|v|\leq P(t)}\,dv +||n||_\infty\leq
\pi ||f^{in}||_\infty P^2(t)+||n||_\infty.
\end{align*}
The third estimate is similarly proved.
For $(iv)$, we use the measure preserving property of the map
$(x,v)\mapsto(X(0),V(0))$ as follows:
\begin{align*}
|D_1(t,x)|&\leq\int_{-\infty}^\infty|\rho(t,y)|\,dy\cr
&\leq\int_{-\infty}^\infty\int_{\bbr^2} f(t,y,v)\,dv\,dy +||n||_1= ||f^{in}||_1+||n||_1.
\end{align*}
\end{proof}
%
%
\begin{lemma}
Suppose that $D^*$ and $B^*$ are $\mathcal{C}^1$ functions on a time interval $I$,
and consider the equation \eqref{4.4}.
Then $k_i$, $i=0,1,2$, are well defined and $\mathcal{C}^1$ on $I$ and satisfy
the following estimates on $I\times\bbr$:
\begin{align*}
(i)&\quad |k_0(t,x)|\leq \pi||f^{in}||_\infty P^2(t)+||n||_\infty.\cr
(ii)&\quad |k_1(t,x)|\leq \pi||f^{in}||_\infty P^2(t).\cr
(iii)&\quad |k_2(t,x)|\leq \pi||f^{in}||_\infty P^2(t).
\end{align*}
\end{lemma}
\begin{proof}
In \eqref{4.4}, we can see that
\[
|k_0(t,x)|\leq|\rho(t,x)|,\quad
|k_1(t,x)|\leq|j_1(t,x)|,\quad
|k_2(t,x)|\leq|j_2(t,x)|.
\]
Hence, the lemma is proved by Lemma 5.1.
\end{proof}
In the next lemma, we use the integral formulae \eqref{5.6}--\eqref{5.7} and Lemma 5.2 to estimate $D_2$ and $B$ fields.
Note that $\alpha=\theta_2-\theta_B$ and $\beta=\theta_2+\theta_B$.
%
%
\begin{lemma}
Suppose that $D^*$ and $B^*$ are $\mathcal{C}^1$ functions on a time interval $I$.
Then the system \eqref{4.2}--\eqref{4.3} has unique $\mathcal{C}^1$ solutions on $I\cap [0,T_1)$,
which is defined by
\[
T_1:=\sup\left\{t: ||\theta_2^{in}||_\infty+||\theta_B^{in}||_\infty
+\int_0^t ||n||_\infty+3\pi ||f^{in}||_\infty P^2(\tau)\,d\tau<\frac{\pi}{2}\right\},
\]
and $D_2$ and $B$ satisfy the following estimates on $I\cap [0,T_1)\times\bbr$:
\[
\begin{aligned}
&|D_2(t,x)|\leq \sqrt{1+(||f^{in}||_1+||n||_1)^2}\tan\left(
||\theta_2^{in}||_\infty+||\theta_B^{in}||_\infty
+\int_0^t ||n||_\infty+3\pi ||f^{in}||_\infty P^2(\tau)\,d\tau\right),\\
&|B(t,x)|\leq \tan\left(
||\theta_2^{in}||_\infty+||\theta_B^{in}||_\infty
+\int_0^t ||n||_\infty+3\pi ||f^{in}||_\infty P^2(\tau)\,d\tau\right).
\end{aligned}
\]
\end{lemma}
\begin{proof}
The existence interval $[0,T_1)$ is given by Lemma 5.2 followed by Corollary 3.1.
Since $\theta_2=(\alpha+\beta)/2$ and $\theta_B=(\beta-\alpha)/2$, we have the following
estimates from \eqref{5.6}--\eqref{5.7} and Lemma 5.2:
\[
\begin{aligned}
|\theta_2(t,x)|&\leq \frac{|\alpha^{in}(\xi(0))|+|\beta^{in}(\eta(0))|}{2}
+\int_0^t||k_0(\tau)||_\infty+||k_1(\tau)||_\infty+||k_2(\tau)||_\infty\,d\tau\\
&\leq ||\theta_2^{in}||_\infty+||\theta_B^{in}||_\infty
+\int_0^t||k_0(\tau)||_\infty+||k_1(\tau)||_\infty+||k_2(\tau)||_\infty\,d\tau\\
&\leq ||\theta_2^{in}||_\infty+||\theta_B^{in}||_\infty
+\int_0^t ||n||_\infty+3\pi ||f^{in}||_\infty P^2(\tau)\,d\tau,
\end{aligned}
\]
and similarly we have for $\theta_B$
\[
|\theta_B(t,x)|\leq ||\theta_2^{in}||_\infty+||\theta_B^{in}||_\infty
+\int_0^t ||n||_\infty+3\pi ||f^{in}||_\infty P^2(\tau)\,d\tau.
\]
Hence, the lemma is proved by \eqref{4.6} as follows:
\[
\begin{aligned}
|D_2(t,x)|&\leq\sqrt{1+D_1^2(t,x)}\tan(|\theta_2(t,x)|)\\
&\leq\sqrt{1+(||f^{in}||_1+||n||_1)^2}\tan\left(
||\theta_2^{in}||_\infty+||\theta_B^{in}||_\infty
+\int_0^t ||n||_\infty+3\pi ||f^{in}||_\infty P^2(\tau)\,d\tau\right),
\end{aligned}
\]
and
\[
\begin{aligned}
|B(t,x)|&\leq\tan(|\theta_2(t,x)|)\\
&\leq\tan\left(
||\theta_2^{in}||_\infty+||\theta_B^{in}||_\infty
+\int_0^t ||n||_\infty+3\pi ||f^{in}||_\infty P^2(\tau)\,d\tau\right).
\end{aligned}
\]
\end{proof}
%
%
\begin{remark}
By the same argument with the proof of Lemma 5.3,
we can see from \eqref{5.6}--\eqref{5.7}
that $\alpha$, $\beta$, $\theta_2$, and $\theta_B$ are bounded by the same quantity as follows:
\[
\max\{|\alpha(t,x)|,|\beta(t,x)|\}\leq||\theta_2^{in}||_\infty+||\theta_B^{in}||_\infty
+\int_0^t ||n||_\infty+3\pi ||f^{in}||_\infty P^2(\tau)\,d\tau.
\]
Since $\alpha=\theta_2-\theta_B$ and $\beta=\theta_2+\theta_B$, the above inequality gives the following estimate:
\[
|\theta_2(t,x)|+|\theta_B(t,x)|\leq||\theta_2^{in}||_\infty+||\theta_B^{in}||_\infty
+\int_0^t ||n||_\infty+3\pi ||f^{in}||_\infty P^2(\tau)\,d\tau.
\]
In other words, on $[0,T_1)$ we have
\[
\max\Big\{|\theta_2(t,x)|,|\theta_B(t,x)|,|\alpha(t,x)|,|\beta(t,x)|\Big\}
\leq |\theta_2(t,x)|+|\theta_B(t,x)|<\frac{\pi}{2}.
\]
\end{remark}
%
%
\begin{lemma}
Suppose that $D^*$ and $B^*$ are $\mathcal{C}^1$ functions on $[0,T_1)$,
and consider the following Vlasov equation for $f^\dag$:
\begin{align*}
\partial_tf^\dag+\hat{v}_1\partial_xf^\dag+(D_1+\hat{v}_2B,D_2-\hat{v}_1B)\cdot\nabla_vf^\dag=0,
\end{align*}
where $D$ and $B$ are given by the system \eqref{5.1} and \eqref{4.2}--\eqref{4.6},
with the same initial data $f^\dag(0)=f^{in}$.
Then its momentum support $P^\dag(t)$ satisfies the following integral inequality on $[0,T_1)$:
\[
\begin{aligned}
&P^\dag(t)\leq P+2\sqrt{1+(||f^{in}||_1+||n||_1)^2}\\
&\hspace{3cm}\times\int_0^t
\cos^{-1}\left(||\theta_2^{in}||_\infty+||\theta_B^{in}||_\infty
+\int_0^s||n||_\infty+3\pi||f^{in}||_\infty P^2(\tau)\,d\tau\right)\,ds,
\end{aligned}
\]
where $P(t)$ is the momentum support of $f$ in \eqref{5.1}.
\end{lemma}
\begin{proof}
By Lemma 5.1 and 5.3, we have
\begin{align*}
&|D_1(t,x)|\leq ||f^{in}||_1+||n||_1,\\
&|D_2(t,x)|\leq\sqrt{1+(||f^{in}||_1+||n||_1)^2}\tan\left(
||\theta_2^{in}||_\infty+||\theta_B^{in}||_\infty
+\int_0^t ||n||_\infty+3\pi ||f^{in}||_\infty P^2(\tau)\,d\tau\right),\\
&|B(t,x)|\leq\tan\left(
||\theta_2^{in}||_\infty+||\theta_B^{in}||_\infty
+\int_0^t ||n||_\infty+3\pi ||f^{in}||_\infty P^2(\tau)\,d\tau\right).
\end{align*}
Hence, we obtain
\[
\begin{aligned}
|D(t,x)|^2&\leq\Big(1+(||f^{in}||_1+||n||_1)^2\Big)\Big(1+\tan^2\Big(\cdots\Big)\Big)\\
&=\Big(1+(||f^{in}||_1+||n||_1)^2\Big)\cos^{-2}\Big(\cdots\Big),
\end{aligned}
\]
and therefore,
\[
||D(t)||_\infty+||B(t)||_\infty\leq 2\sqrt{1+(||f^{in}||_1+||n||_1)^2}\cos^{-1}\Big(\cdots\Big).
\]
We now use the inequality \eqref{5.8} for $f^\dag$ and $P^\dag$ to get the desired result.
\end{proof}
%
%
\subsection{A priori estimates for derivatives of the solutions}
In this part, we estimate the derivatives of $f$, $\alpha$, and $\beta$.
In Lemma 5.3, we introduced a finite interval $[0,T_1)$ on which we could estimate
the field quantities. To estimate their derivatives, we need an additional interval $[0,T_2)$
which is defined as follows:
\[
T_2:=\sup\left\{t:||\theta_2^{in}||_\infty+||\theta_B^{in}||_\infty+\int_0^t||n||_\infty+3\pi||f^{in}||_\infty
P^2(s)\,ds<\arctan\frac{1}{P(t)}\right\}.
\]
Note that $T_2>0$ because of the assumption on initial data A3. We remind that the assumption A3 gives the separation between
the Vlasov and the Born-Infeld characteristics at $t=0$. Hence, we can see that the interval $[0,T_2)$ is the maximal
interval on which the Vlasov and the Born-Infeld characteristics are separated: since $\alpha$ and $\beta$ are bounded by
the LHS of the inequality in the definition of $T_2$ (see Remark 5.1), we have on $[0,T_2)$
\[
\max\{|\alpha(t,x)|,|\beta(t,x)|\}<\arctan\frac{1}{P(t)},
\]
which is equivalent to
\[
\hat{P}(t)<\cos(\max\{|\alpha(t,x)|,|\beta(t,x)|\})
=\min\{\cos\alpha(t,x),\cos\beta(t,x)\}.
\]
Hence, we have
\[
-\cos\beta(t,x)<\hat{v}_1<\cos\alpha(t,x)
\]
for any $x$ and $v$ satisfying $f(t,x,v)\neq 0$.

We now fix a closed and bounded interval $[0,T]\subset [0,T_1)\cap [0,T_2)$,
and all the following estimates in this subsection will be
considered on $[0,T]$. Note that $P(T)$ is finite, and $C_T$ will denote a positive
constant depending only on $T$ and $P(T)$ and vary from line to line.
We first estimate $x$ and $v$ derivatives of $f$.
%
%
\begin{lemma}
Suppose that $D^*$ and $B^*$ are $\mathcal{C}^1$ functions on $[0,T]$,
and consider the following Vlasov equation for $f^\dag$:
\begin{align*}
\partial_tf^\dag+\hat{v}_1\partial_xf^\dag+(D_1+\hat{v}_2B,D_2-\hat{v}_1B)\cdot\nabla_vf^\dag=0,
\end{align*}
where $D$ and $B$ are given by the system \eqref{5.1} and \eqref{4.2}--\eqref{4.6}, with the same initial
data $f^\dag(0)=f^{in}$. Then we have the following estimate for $\nabla_{x,v}f^\dag$:
\[
||\nabla_{x,v}f^\dag(t)||_\infty\leq ||\nabla_{x,v}f^{in}||_\infty
+C_T\int_0^t\left(1+||\partial_x\alpha(s)||_\infty
+||\partial_x\beta(s)||_\infty\right)||\nabla_{x,v}f^\dag(s)||_\infty\,ds,
\]
where $C_T$ is a positive constant depending on $T$, $P(T)$, and initial data.
\end{lemma}
\begin{proof}
By direct differentiation with respect to $x$ and $v$, we obtain
\begin{align*}
&\partial_t(\partial_xf^\dag)+\hat{v}_1\partial_x(\partial_xf^\dag)
+(D+(\hat{v}_2,-\hat{v}_1)B)\cdot\nabla_v(\partial_xf^\dag)\\
&\hspace{5cm}=-(\partial_xD+(\hat{v}_2,-\hat{v}_1)\partial_xB)\cdot\nabla_vf^\dag,\\
&\partial_t(\partial_{v_i}f^\dag)+\hat{v}_1\partial_x(\partial_{v_i}f^\dag)
+(D+(\hat{v}_2,-\hat{v}_1)B)\cdot\nabla_v(\partial_{v_i}f^\dag)\\
&\hspace{5cm}=-(\partial_{v_i}\hat{v}_1)(\partial_xf^\dag)-(\partial_{v_i}\hat{v}_2,-\partial_{v_i}\hat{v}_1)B\cdot\nabla_vf^\dag,
\end{align*}
where $i=1,2$, and
by taking the characteristic curve \eqref{5.2}, we have
\[
||\nabla_{x,v}f^\dag(t)||_\infty\leq ||\nabla_{x,v}f^{in}||_\infty+C_T\int_0^t\left(1+||\partial_xD(s)||_\infty
+||\partial_xB(s)||_\infty\right)||\nabla_{x,v}f^\dag(s)||_\infty\,ds,
\]
where we used $|B|\leq C_T$ by Lemma 5.3.
By \eqref{4.5} and Lemma 5.1, we obtain
\begin{equation}\label{5.9}
|\partial_xD_1(t,x)|\leq|\rho(t,x)|\leq C_T.
\end{equation}
We take $x$ derivative on \eqref{4.6} to have
\[
\partial_xD_2=\frac{D_1\partial_xD_1}{\sqrt{1+D_1^2}}\tan\theta_2
+\sqrt{1+D_1^2}\frac{\partial_x\theta_2}{\cos^2\theta_2}.
\]
Since $|\theta_2(t,x)|<\frac{\pi}{2}$ on $[0,T_1)$, we obtain
\begin{equation}\label{5.10}
|\partial_xD_2(t,x)|\leq C_T(1+|\partial_x\theta_2(t,x)|)\leq C_T(1+||\partial_x\alpha(t)||_\infty+||\partial_x\beta(t)||_\infty).
\end{equation}
Similarly, we obtain
\begin{equation}\label{5.11}
|\partial_xB(t,x)|\leq C_T(||\partial_x\alpha(t)||_\infty+||\partial_x\beta(t)||_\infty),
\end{equation}
and this completes the proof.
\end{proof}
We now estimate the $x$ derivatives of $\alpha$ and $\beta$. In the next lemma,
we introduce new quantities in order to estimate $\partial_x\alpha$ and $\partial_x\beta$ as in \cite{Li},
but the lemma is proved only when the Vlasov and the Born-Infeld
characteristics are well separated. We will see in the proof that it is necessary
to introduce the interval $[0,T_2)$ and restrict all the arguments to it.
%
%
\begin{lemma}
Suppose that $D^*$ and $B^*$ are $\mathcal{C}^1$ functions on $[0,T]$,
and consider the following quantities $u$ and $w$:
\[
u:=(-\cos\alpha-\cos\beta)\partial_x\alpha\quad\mbox{and}\quad
w:=(\cos\alpha+\cos\beta)\partial_x\beta.
\]
Then, they are bounded on $[0,T]$
by a positive constant $C_T$ which depends only on $T$, $P(T)$, and initial data.
\end{lemma}
\begin{proof}
By direct calculations, we have
\begin{align*}
&\partial_tu=(\sin\alpha\partial_t\alpha+\sin\beta\partial_t\beta)\partial_x\alpha
+(-\cos\alpha-\cos\beta)\partial_x\partial_t\alpha,\\
&\partial_xu=(\sin\alpha\partial_x\alpha+\sin\beta\partial_x\beta)\partial_x\alpha
+(-\cos\alpha-\cos\beta)\partial_x^2\alpha,
\end{align*}
and
\begin{equation}\label{5.12}
\begin{aligned}
&\partial_t\partial_x\alpha+\sin\beta\partial_x\beta\partial_x\alpha-\cos\beta\partial_x^2\alpha\\
&=(\partial_xk_0\sin\theta_B+k_0\cos\theta_B\partial_x\theta_B
+\partial_xk_1\sin\theta_2+k_1\cos\theta_2\partial_x\theta_2
+\partial_xk_2\cos\theta_2-k_2\sin\theta_2\partial_x\theta_2)\cos\theta_2\\
&\hspace{0.5cm}-(k_0\sin\theta_B+k_1\sin\theta_2+k_2\cos\theta_2)\sin\theta_2\partial_x\theta_2.
\end{aligned}
\end{equation}
We now derive an equation for $u$ as follows:
\begin{align*}
&\partial_tu-\cos\beta\partial_xu\\
&=\sin\alpha\partial_t\alpha\partial_x\alpha+\sin\beta\partial_t\beta\partial_x\alpha
-\cos\alpha\partial_x\partial_t\alpha-\cos\beta\partial_x\partial_t\alpha\\
&\hspace{0.5cm}-\cos\beta\sin\alpha(\partial_x\alpha)^2-\cos\beta\sin\beta\partial_x\beta\partial_x\alpha
+\cos\alpha\cos\beta\partial_x^2\alpha+\cos^2\beta\partial_x^2\alpha\\
&=\sin\alpha\partial_x\alpha(\partial_t\alpha-\cos\beta\partial_x\alpha)+\sin\beta\partial_x\alpha(\partial_t\beta
+\cos\alpha\partial_x\beta)\\
&\hspace{0.5cm}-(\cos\alpha+\cos\beta)(\partial_t\partial_x\alpha+\sin\beta\partial_x\beta\partial_x\alpha-\cos\beta\partial_x^2\alpha).
\end{align*}
Since the inhomogeneous terms of \eqref{4.2} and \eqref{4.3} are same, we have
\begin{align*}
&\partial_tu-\cos\beta\partial_xu\\
&=(\sin\alpha+\sin\beta)\partial_x\alpha(k_0\sin\theta_B+k_1\sin\theta_2+k_2\cos\theta_2)\cos\theta_2\\
&\hspace{0.5cm}-(\cos\alpha+\cos\beta)(\partial_xk_0\sin\theta_B+\partial_xk_1\sin\theta_2+\partial_xk_2\cos\theta_2)\cos\theta_2\\
&\hspace{0.5cm}-(\cos\alpha+\cos\beta)(k_0\cos\theta_B\partial_x\theta_B
+k_1\cos\theta_2\partial_x\theta_2-k_2\sin\theta_2\partial_x\theta_2)\cos\theta_2\\
&\hspace{0.5cm}+(\cos\alpha+\cos\beta)(k_0\sin\theta_B+k_1\sin\theta_2+k_2\cos\theta_2)\sin\theta_2\partial_x\theta_2\\
&=: I_1+I_2+I_3+I_4,
\end{align*}
where we used \eqref{5.12}. For $I_1$, we have
\begin{align*}
I_1&=-\frac{\sin\alpha+\sin\beta}{\cos\alpha+\cos\beta}u(k_0\sin\theta_B+k_1\sin\theta_2+k_2\cos\theta_2)\cos\theta_2\\
&=-u(k_0\sin\theta_2\sin\theta_B+k_1\sin^2\theta_2+k_2\cos\theta_2\sin\theta_2)\\
&=-\frac{u}{2}\Big(k_0(\cos\alpha-\cos\beta)+k_1(1-\cos(2\theta_2))+k_2\sin(2\theta_2)\Big),
\end{align*}
where we used trigonometric identities. For $I_2$, we have
\begin{equation}\label{5.13}
\begin{aligned}
I_2&=-2\cos\theta_2\cos\theta_B(\partial_xk_0\sin\theta_B+\partial_xk_1\sin\theta_2+\partial_xk_2\cos\theta_2)\cos\theta_2\\
&=-2(\cos^2\theta_2\cos\theta_B\sin\theta_B\partial_xk_0
+\cos^2\theta_2\sin\theta_2\cos\theta_B\partial_xk_1
+\cos^3\theta_2\cos\theta_B\partial_xk_2).
\end{aligned}
\end{equation}
For $I_3$ and $I_4$, we note that
\[
\frac{u+w}{2}=(\cos\alpha+\cos\beta)\partial_x\theta_B\quad\mbox{and}\quad
\frac{w-u}{2}=(\cos\alpha+\cos\beta)\partial_x\theta_2.
\]
Hence we have
\begin{align*}
&I_3=-\left(k_0\cos\theta_B\frac{u+w}{2}+k_1\cos\theta_2\frac{w-u}{2}-k_2\sin\theta_2\frac{w-u}{2}\right)\cos\theta_2,\\
&I_4=\frac{w-u}{2}(k_0\sin\theta_B+k_1\sin\theta_2+k_2\cos\theta_2)\sin\theta_2,
\end{align*}
and then $I_3+I_4$ is written as follows:
\begin{align*}
&I_3+I_4\\
&=\frac{u}{2}\Big(-k_0\cos\alpha+k_1\cos(2\theta_2)-k_2\sin(2\theta_2)\Big)
+\frac{w}{2}\Big(-k_0\cos\beta-k_1\cos(2\theta_2)+k_2\sin(2\theta_2)\Big).
\end{align*}
We write $I_1+I_3+I_4$ as follow for simplicity:
\begin{align*}
&I_1+I_3+I_4\\
&=u\left(k_0\left(-\cos\alpha+\frac{1}{2}\cos\beta\right)
+k_1\left(-\frac{1}{2}+\cos(2\theta_2)\right)
-k_2\sin(2\theta_2)\right)\\
&\hspace{0.5cm}+w\left(-\frac{1}{2}k_0\cos\beta-\frac{1}{2}k_1\cos(2\theta_2)+\frac{1}{2}k_2\sin(2\theta_2)\right)\\
&=:K_1u+K_2w,
\end{align*}
where $K_1$ and $K_2$ are elementary functions of $k_0$, $k_1$, $k_2$, $\alpha$, and $\beta$.
We now control the derivative terms $\partial_xk_i$, $i=0,1,2$, in $I_2$.
By direct calculations, we have
\begin{equation}\label{5.14}
\begin{aligned}
\partial_xk_0&=-\frac{D_1}{1+D_1^2}\partial_x\rho-\rho\left(\frac{1}{(1+D_1^2)^2}(\partial_xD_1(1+D_1^2)
-2D_1^2\partial_xD_1)\right)\\
&=-\frac{D_1}{1+D_1^2}\partial_x\rho-\rho^2\frac{1-D_1^2}{(1+D_1^2)^2},
\end{aligned}
\end{equation}
and similarly,
\begin{equation}\label{5.15}
\partial_xk_1=\frac{D_1}{1+D_1^2}\partial_xj_1+j_1\rho\frac{1-D_1^2}{(1+D_1^2)^2},
\end{equation}
\begin{equation}\label{5.16}
\partial_xk_2=-\frac{1}{\sqrt{1+D_1^2}}\partial_xj_2+j_2\rho\frac{D_1}{(1+D_1^2)^{3/2}}.
\end{equation}
If we plug the above results into \eqref{5.13}, then we can see that $I_2$ is written as follows:
\[
I_2=\Omega_0\partial_x\rho+\Omega_1\partial_xj_1+\Omega_2\partial_xj_2+L_0+L_1+L_2,
\]
where $\Omega_i$, $i=0,1,2$, are elementary functions of $\alpha$, $\beta$, and $D_1$, while
$L_i$, $i=0,1,2$, are elementary functions of $\alpha$, $\beta$, $\rho$, $j_1$, $j_2$, and $D_1$.
Consequently, we obtain the following equation for $u$:
\begin{equation}\label{5.17}
\partial_tu-\cos\beta\partial_xu
=K_1u+K_2w+\Omega_0\partial_x\rho+\Omega_1\partial_xj_1+\Omega_2\partial_xj_2+L_0+L_1+L_2.
\end{equation}
By Lemma 5.3, the system \eqref{4.2}--\eqref{4.3} has unique $\mathcal{C}^1$ solutions on $[0,T]$,
hence the characteristic curve \eqref{5.4}--\eqref{5.5} are well defined on $[0,T]$.
Along the curve \eqref{5.4}, we can rewrite \eqref{5.14} as follows:
\begin{equation}\label{5.18}
\begin{aligned}
u(t,x)&=u^{in}(\xi(0))+\int_0^t(K_1u+K_2w+L_0+L_1+L_2)(\tau,\xi(\tau))\,d\tau\\
&\hspace{0.5cm}+\int_0^t(\Omega_0\partial_x\rho
+\Omega_1\partial_xj_1
+\Omega_2\partial_xj_2)(\tau,\xi(\tau))\,d\tau.
\end{aligned}
\end{equation}
For simplicity, the last three term of \eqref{5.15} will be denoted by
\[
J_i=\int_0^t\Omega_i(\tau,\xi(\tau))(\partial_x j_i)(\tau,\xi(\tau))\,d\tau
\]
for $i=0,1,2$ and $j_0=\rho$.\bigskip

To estimate each $J_i$, we use Glassey-Strauss' argument in \cite{GSt86}. We first introduce the
differential operators $T$ and $S$:
\[
T:=\partial_t-\cos\beta\partial_x\quad\mbox{and}\quad
S:=\partial_t+\hat{v}_1\partial_x,
\]
by which $x$ derivatives can be written as
\[
\partial_x=\frac{S-T}{\hat{v}_1+\cos\beta}.
\]
Note that the denominator above does not vanish on $[0,T]$ because we have $-\cos\beta(t,x)<\hat{v}_1<\cos\alpha(t,x)$
on $[0,T_2)$ for $x$ and $v$ satisfying $f(t,x,v)\neq 0$.
In other words, the assumption A3 and the restriction to $[0,T_2)$ enable us to apply the argument of \cite{GSt86}.
$J_0$ term is estimated as follows:
\begin{align*}
J_0&=\int_0^t\Omega_0(\tau,\xi(\tau))(\partial_x\rho)(\tau,\xi(\tau))\,d\tau\\
&=\int_0^t\int_{\bbr^2}\Omega_0(\tau,\xi(\tau))(\partial_xf)(\tau,\xi(\tau),v)\,dv\,d\tau
-\int_0^t\Omega_0(\tau,\xi(\tau))(\partial_xn)(\xi(\tau))\,d\tau\\
&=: J_{01}+J_{02}.
\end{align*}
We rewrite $J_{01}$ in terms of $T$ and $S$ as follows:
\begin{eqnarray}
J_{01}&=&\int_0^t\int_{\bbr^2}\Omega_0(\tau,\xi(\tau))
\left(\frac{S-T}{\hat{v}_1+\cos\beta}\right)f(\tau,\xi(\tau),v)\,dv\,d\tau\cr
&=&\int_0^t\int_{\bbr^2}\Omega_0(\tau,\xi(\tau))
\frac{1}{\hat{v}_1+\cos\beta(\tau,\xi(\tau))}(Sf)(\tau,\xi(\tau),v)\,dv\,d\tau\cr
&&{}-\int_0^t\int_{\bbr^2}\Omega_0(\tau,\xi(\tau))
\frac{1}{\hat{v}_1+\cos\beta(\tau,\xi(\tau))}(Tf)(\tau,\xi(\tau),v)\,dv\,d\tau\cr
&=&\int_0^t\int_{\bbr^2}
\frac{\Omega_0(\tau,\xi(\tau))}{\hat{v}_1+\cos\beta(\tau,\xi(\tau))}
\nabla_v\cdot\Big[(-D-(\hat{v}_2,-\hat{v}_1)B)f(\tau,\xi(\tau),v)\Big]\,dv\,d\tau\cr
&&{}-\int_0^t\int_{\bbr^2}
\frac{\Omega_0(\tau,\xi(\tau))}{\hat{v}_1+\cos\beta(\tau,\xi(\tau))}
\partial_\tau\Big[f(\tau,\xi(\tau),v)\Big]\,dv\,d\tau\cr
&=&\int_0^t\int_{\bbr^2}
\nabla_v\left[\frac{\Omega_0(\tau,\xi(\tau))}{\hat{v}_1+\cos\beta(\tau,\xi(\tau))}\right]\cdot
(D+(\hat{v}_2,-\hat{v}_1)B)f(\tau,\xi(\tau),v)\,dv\,d\tau\label{5.19}\\
&&{}+\int_0^t\int_{\bbr^2}
\partial_\tau\left[\frac{\Omega_0(\tau,\xi(\tau))}{\hat{v}_1+\cos\beta(\tau,\xi(\tau))}\right]
f(\tau,\xi(\tau),v)\,dv\,d\tau\label{5.20}\\
&&{}-\int_{\bbr^2}\left(\frac{\Omega_0(t,x)}{\hat{v}_1+\cos\beta(t,x)}f(t,x,v)-
\frac{\Omega_0(0,\xi(0))}{\hat{v}_1+\cos\beta^{in}(\xi(0))}f^{in}(\xi(0),v)\right)\,dv.\label{5.21}
\end{eqnarray}
We remind that
\[
\Omega_0=2\cos^2\theta_2\cos\theta_B\sin\theta_B\frac{D_1}{1+D_1^2},
\]
and we need the followings, which can be verified by direct calculations:
\[
\partial_{v_i}\left[\frac{1}{\hat{v}_1+\cos\beta}\right]
=-\frac{1}{\sqrt{1+|v|^2}}\frac{1}{(\hat{v}_1+\cos\beta)^2}\left(\delta_{1i}
-\frac{v_1 v_i}{1+|v|^2}\right),
\]
\[
\partial_\tau\Big[\theta_2(\tau,\xi(\tau))\Big]=
-\frac{1}{2}w(\tau,\xi(\tau)),
\]
\[
\partial_\tau\Big[\theta_B(\tau,\xi(\tau))\Big]=
-\frac{1}{2}w(\tau,\xi(\tau))+ inhom,
\]
\[
\partial_\tau\Big[D_1(\tau,\xi(\tau))\Big]=
-j_1(\tau,\xi(\tau))-\cos\beta(\tau,\xi(\tau))\rho(\tau,\xi(\tau)),
\]
where `$inhom$' is the inhomogeneous term of \eqref{4.2} evaluated at $(\tau,\xi(\tau))$.
Note that $D$, $B$, $\rho$, $j$, $k_0$, $k_1$, and $k_2$ are bounded by a constant $C_T$
on $[0,T]$ by Lemma 5.1--5.3.
By applying the above results to \eqref{5.16}--\eqref{5.18}, we obtain the estimate for $J_{01}$.
\begin{eqnarray*}
|J_{01}|&\leq&C_T\int_0^t\int_{|v|\leq P(\tau)}\frac{1}{|\hat{v}_1+\cos\beta(\tau,\xi(\tau))|^2}\,dv\,d\tau\cr
&&{}+C_T\int_0^t\int_{|v|\leq P(\tau)}\frac{|w(\tau,\xi(\tau))|}{|\hat{v}_1+\cos\beta(\tau,\xi(\tau))|^2}\,dv\,d\tau\cr
&&{}+C\int_{|v|\leq P(t)}\frac{1}{|\hat{v}_1+\cos\beta(t,x)|}\,dv
+C\cr
&\leq& C_T+C_T\int_0^t||w(\tau)||_\infty\,d\tau.
\end{eqnarray*}
In the last inequality, we used the fact that $-\cos\beta(s,x)<\hat{v}_1<\cos\alpha(s,x)$ on $[0,T_3)$ for any $x$ and $v$ satisfying
$f(s,x,v)\neq 0$. $J_{02}$ is clearly bounded by $C_T$: $|J_{02}|\leq C_T$, and $J_1$ and $J_2$ are estimated
by the same way:
\[
|J_i|\leq C_T+C_T\int_0^t||w(\tau)||_\infty\,d\tau,\qquad i=1,2.
\]
The other terms, $K_1$, $K_2$, $L_0$, $L_1$, and $L_2$ in \eqref{5.14} are easily bounded by $C_T$ on $[0,T]$.
Therefore, we obtain the following estimate for $u$ from \eqref{5.15}:
\[
||u(t)||_\infty\leq C_T+C_T\int_0^t||u(\tau)||_\infty+||w(\tau)||_\infty\,d\tau.
\]
In the same way, we have the following estimate for $w$ after long calculations:
\[
||w(t)||_\infty\leq C_T+C_T\int_0^t||u(\tau)||_\infty+||w(\tau)||_\infty\,d\tau.
\]
We apply the Gr{\"o}nwall inequality to obtain the desired result,
and this completes the proof.
\end{proof}
%
%
\begin{lemma}
Suppose that $D^*$ and $B^*$ are $\mathcal{C}^1$ functions on $[0,T]$,
and consider the following quantities:
\[
||\partial_x\alpha(t)||_\infty,\quad
||\partial_x\beta(t)||_\infty,\quad
||\partial_xD(t)||_\infty,\quad
||\partial_xB(t)||_\infty,\quad
||\nabla_{x,v}f^\dag(t)||_\infty,
\]
where $f^\dag$ is given by Lemma 5.5.
Then, they are bounded on $[0,T]$ by a positive constant $C_T$
which depends only on $T$, $P(T)$, and initial data.
\end{lemma}
\begin{proof}
By Lemma 5.6, we have the following estimate on $[0,T]$:
\[
C_T\geq(\cos\alpha(t,x)+\cos\beta(t,x))|\partial_x\alpha(t,x)|
=2\cos\theta_2(t,x)\cos\theta_B(t,x)|\partial_x\alpha(t,x)|.
\]
On the other hand, we have on $[0,T_1)$ (see Remark 5.1)
\[
|\theta_2(t,x)|<\frac{\pi}{2},\quad|\theta_B(t,x)|<\frac{\pi}{2}.
\]
Therefore, we obtain on $[0,T]$
\[
\cos\theta_2(t,x)\geq\frac{1}{C_T},\quad\cos\theta_B(t,x)\geq\frac{1}{C_T},
\]
and this gives $||\partial_x\alpha(t)||_\infty\leq C_T$.
Similarly, we obtain the boundedness of $||\partial_x\beta(t)||_\infty$.
The boundedness of $\partial_xD$ and $\partial_xB$ is obtained by \eqref{5.9}--\eqref{5.11} as follows:
\[
||\partial_xD(t)||_\infty+||\partial_xB(t)||_\infty
\leq C_T(1+||\partial_x\alpha(t)||_\infty+||\partial_x\beta(t)||_\infty)\leq C_T.
\]
By using Gr{\"o}nwall's inequality on the result of Lemma 5.5, we can see that
$||\nabla_{x,v}f^\dag(t)||_\infty$ is bounded by $C_T$.
This completes the proof.
\end{proof}
%
%
\section{Iteration scheme and proof of convergence}\setcounter{equation}{0}
In this section, we prove the main theorem.
We first introduce the iteration scheme for \eqref{4.1}--\eqref{4.6} and show that the sequence
converges to a $\mathcal{C}^1$ function by using the a priori estimates obtained in
the previous section.
%
%
\subsection{Iteration}
We use a standard iteration.
We take sequences $f^{n}$, $\theta_2^{n}$, and $\theta_B^{n}$, $n=1,2,\cdots$, as follows.
Define $f^{0}:=f^{in}$, $\theta_2^{0}:=\theta_2^{in}$, and $\theta_B^{0}:=\theta_B^{in}$, or equivalently
$\alpha^0:=\alpha^{in}=\theta_2^{in}-\theta_B^{in}$ and $\beta^0:=\beta^{in}=\theta_2^{in}+\theta_B^{in}$.
For given $(n-1)$-th
step, we define $n$-th step as follows: $f^{n}$ is taken as the solution of the following equation:
\begin{equation}\label{6.1}
\begin{aligned}
&\partial_tf^{n}+\hat{v}_1\partial_xf^{n}+(D^{n-1}_1+\hat{v}_2B^{n-1},D^{n-1}_2-\hat{v}_1B^{n-1})\cdot\nabla_vf^{n}=0,\\
&f^{n}(0,x,v)=f^{in}(x,v),
\end{aligned}
\end{equation}
from which we take $\rho^{n}$, $j^{n}$, and $D_1^{n}$ as follows:
\begin{equation}\label{6.2}
\begin{aligned}
&\rho^{n}(t,x)=\int_{\bbr^2}f^{n}(t,x,v)\,dv-n(x)\quad\mbox{and}\quad D^{n}_1(t,x)=\int_{-\infty}^x\rho^{n}(t,y)\,dy,\\
&j^{n}(t,x)=\int_{\bbr^2}\hat{v}f^{n}(t,x,v)\,dv,
\end{aligned}
\end{equation}
and then we get $k_i^{n}$, $i=0,1,2$, from $\rho^n$, $j^n$, and $D_1^n$ by using \eqref{4.4}.
$\theta^{n}_2$ and $\theta^{n}_B$ are taken as the solution of
the following inhomogeneous quasilinear hyperbolic system:
\begin{equation}\label{6.3}
\begin{aligned}
&\partial_t\alpha^{n}-(\cos\beta^{n})\partial_x\alpha^{n}=\cos\theta_2^{n}(k_0^{n}\sin\theta_B^{n}+k_1^{n}\sin\theta_2^{n}+k_2^{n}\cos\theta_2^{n}),\\
&\partial_t\beta^n+(\cos\alpha^n)\partial_x\beta^n=\cos\theta_2^n(k_0^n\sin\theta_B^n+k_1^n\sin\theta_2^n+k_2^n\cos\theta_2^n),\\
&\theta_2^n(0,x)=\theta_2^{in}(x)\quad\mbox{and}\quad \theta_B^n(0,x)=\theta_B^{in}(x),
\end{aligned}
\end{equation}
where $\alpha^{n}=\theta_2^n-\theta_B^n$ and $\beta^n=\theta_2^n+\theta_B^n$.
Finally, $D_2^n$ and $B^n$ are obtained:
\begin{equation}\label{6.4}
D_2^n=\sqrt{1+(D_1^n)^2}\tan\theta_2^n\quad\mbox{and}\quad
B^n=\tan\theta_B^n,
\end{equation}
and this completes the $n$-th step.
%
%
\subsection{Some remarks} We remark that the a priori estimates obtained
in the previous section can be applied to each $n$-th step uniformly on $n$.
In the previous section, we first defined the momentum support
and then obtained an integral inequality for the momentum support in Lemma 5.4.
We now redefine $P(t)$ as the solution
of the following integral equation:
\begin{equation}\label{6.5}
\begin{aligned}
&P(t)= P+2\sqrt{1+(||f^{in}||_1+||n||_1)^2}\\
&\hspace{3cm}\times\int_0^t
\cos^{-1}\left(||\theta_2^{in}||_\infty+||\theta_B^{in}||_\infty
+\int_0^s||n||_\infty+3\pi||f^{in}||_\infty P^2(\tau)\,d\tau\right)\,ds,
\end{aligned}
\end{equation}
where $P$, $f^{in}$, $n$, $\theta_2^{in}$, and $\theta_B^{in}$ are given initial conditions satisfying the assumptions A1--A3.
We define a time interval $[0,{\mathfrak T}_0)$ as the maximal interval on which the solution of \eqref{6.5} exists.
Then, $\mathfrak{T}_1$ and $\mathfrak{T}_2$ are also redefined as follows:
\[
\mathfrak{T}_1:=\sup\left\{t:||\theta_2^{in}||_\infty+||\theta_B^{in}||_\infty+\int_0^t||n||_\infty
+3\pi||f^{in}||_\infty P^2(s)\,ds<\frac{\pi}{2}\right\},
\]
\[
\mathfrak{T}_2:=\sup\left\{t:||\theta_2^{in}||_\infty+||\theta_B^{in}||_\infty+\int_0^t||n||_\infty+3\pi||f^{in}||_\infty
P^2(s)\,ds<\arctan\frac{1}{P(t)}\right\}.
\]
We now take
\[
I:=[0,\mathfrak{T}_0)\cap [0,\mathfrak{T}_1)\cap [0,\mathfrak{T}_2).
\]

On the other hand, the momentum support of $f^{n+1}$ is defined as follows:
\[
P_n(t):=\sup\left\{|v|:f^{n+1}(s,x,v)\neq 0,\ 0\leq s\leq t,\ x\in\bbr,\ v\in\bbr^2\right\}.
\]
If we compare the result of Lemma 5.4 with \eqref{6.5} and use the mathematical induction on $n$,
then we can see that each $P_n(t)$ is bounded by $P(t)$, i.e.,
\begin{equation}\label{6.6}
P_n(t)\leq P(t)
\end{equation}
for any $n$. This implies that the existence interval of $P_n$, say $[0,T_{0,n})$, contains that of $P$, i.e.,
\begin{equation}\label{6.7}
[0,\mathfrak{T}_0)\subset [0,T_{0,n}).
\end{equation}

Consider now the iteration functions $f^n$, $\alpha^n$, and $\beta^n$. The $0$-th iteration functions are clearly
$\mathcal{C}^1$ on $I$. Assume that $n$-th iteration functions are of class $\mathcal{C}^1$ on $I$. Then,
$D^n$ and $B^n$ are $\mathcal{C}^1$, and the solutions $X^n$ and $V^n$ of the following characteristic system
exist on $I$,
\begin{equation}\label{6.8}
\begin{aligned}
&\frac{d}{ds}X^n(s)=\hat{V}_1^n(s),\quad X^n(t)=x,\\
&\frac{d}{ds}V^n(s)=D^n(s,X^n(s))+(\hat{V}^n_2(s),-\hat{V}^n_1(s))B^n(s,X^n(s)),\quad V^n(t)=v,
\end{aligned}
\end{equation}
and therefore $f^{n+1}(t,x,v)=f^{in}(X^n(0),V^n(0))$ is a $\mathcal{C}^1$ solution on $I$.
We now obtain $(n+1)$-th step of the quasilinear hyperbolic system \eqref{6.3} with $k^{n+1}_i$, $i=0,1,2$,
which are $\mathcal{C}^1$ on $I$ and satisfy
\[
||k_0^{n+1}(t)||_\infty+||k_1^{n+1}(t)||_\infty+||k_2^{n+1}(t)||_\infty\leq ||n||_\infty+3\pi||f^{in}||_\infty P^2_n(t).
\]
Hence, by Corollary 3.1,
the $(n+1)$-th step of \eqref{6.3} has $\mathcal{C}^1$ solutions on a time interval $[0,T_{1,n})$ which is defined by
\[
T_{1,n}:=\sup\left\{t:||\theta_2^{in}||_\infty+||\theta_B^{in}||_\infty+\int_0^t||n||_\infty
+3\pi||f^{in}||_\infty P^2_n(s)\,ds<\frac{\pi}{2}\right\}.
\]
By the definition of $\mathfrak{T}_1$ and \eqref{6.6}, we can see that
\begin{equation}\label{6.9}
[0,\mathfrak{T}_1)\subset [0,T_{1,n}),
\end{equation}
and therefore
$\alpha^{n+1}$ and $\beta^{n+1}$ are $\mathcal{C}^1$ solutions
on $I$. $D^{n+1}$ and $B^{n+1}$ are clearly $\mathcal{C}^1$ on $I$ due to \eqref{6.2} and \eqref{6.4},
and this completes the $(n+1)$-th step. To summarize, the iteration functions are well defined as $\mathcal{C}^1$ functions
and exist on $I$ uniformly on $n$.

We finally consider the separation between the Vlasov and the Born-Infeld characteristics for each $n$-th step.
A time interval $[0,T_{2,n})$ is defined as we did in Section 5.2.
\[
T_{2,n}:=\sup\left\{t:||\theta_2^{in}||_\infty+||\theta_B^{in}||_\infty+\int_0^t||n||_\infty
+3\pi||f^{in}||_\infty P^2_n(s)\,ds<\arctan\frac{1}{P_n(t)}\right\}.
\]
By \eqref{6.6}, we can see that
\begin{equation}\label{6.10}
[0,\mathfrak{T}_2)\subset [0,T_{2,n}),
\end{equation}
and by the same argument we have
\[
-\cos\beta^{n+1}(t,x)<\hat{v}_1<\cos\alpha^{n+1}(t,x)
\]
for any $x$ and $v$ satisfying $f^{n+1}(t,x,v)\neq 0$ on $[0,\mathfrak{T}_2)$ for any $n$.
Hence, on $[0,\mathfrak{T}_2)$ the Vlasov and the Born-Infeld
characteristics of each $n$-th step are well separated uniformly on $n$.

Consequently, we have \eqref{6.7}, \eqref{6.9}, and \eqref{6.10}, i.e., $I\subset [0,T_{1,n})\cap [0,T_{2,n})\cap [0,T_{3,n})$
for any $n$, and
all the a priori estimates obtained in Section 5 are applied to iteration functions on $I$.
We now fix a closed and bounded interval $[0,T]\subset I$, and on this time interval the following quantities are
bounded by a positive constant $C_T$ which does not depend on $n$:
\begin{equation}\label{6.11}
\nabla^i_{x,v}f^n,\quad
\partial^i_x\rho^n,\quad
\partial^i_xj^n,\quad
\partial^i_x\alpha^n,\quad
\partial^i_x\beta^n,\quad
\partial^i_xD^n,\quad
\partial^i_xB^n,
\end{equation}
where $i=0,1$ and $n\geq 0$ are integers.
%
%
\subsection{Proof of the convergence}
We now show that the iteration functions $\{f^n\}$, $\{\alpha^n\}$, and $\{\beta^n\}$
are Cauchy sequences in $\mathcal{C}^1$. The proof is straightforward and will be given by a sequence of lemmas.
We divide it into two parts: the estimates for the iteration functions themselves and for their derivatives.
Remind that we fixed a compact interval $[0,T]\subset I$, and all the following lemmas and arguments will be
considered only on $[0,T]$. The quantities in \eqref{6.11} will be roughly estimated by $C_T$.
%
%
\subsubsection{Estimates for the iteration functions}
We first show that the sequences $\{f^n\}$, $\{\alpha^n\}$, and $\{\beta^n\}$
converge to continuous functions.
%
%
\begin{lemma}
Consider the iteration functions \eqref{6.1}--\eqref{6.4}.
Then we have
\[
||f^n(t)-f^m(t)||_\infty\leq C_T\int_0^t||D^{m-1}(s)-D^{n-1}(s)||_\infty+
||B^{m-1}(s)-B^{n-1}(s)||_\infty\,ds.
\]
\end{lemma}
\begin{proof}
For any $n$ and $m$, we have
\[
\begin{aligned}
&\partial_tf^{n}+\hat{v}_1\partial_xf^{n}+(D^{n-1}+(\hat{v}_2,-\hat{v}_1)B^{n-1})\cdot\nabla_vf^{n}=0,\\
&\partial_tf^{m}+\hat{v}_1\partial_xf^{m}+(D^{m-1}+(\hat{v}_2,-\hat{v}_1)B^{m-1})\cdot\nabla_vf^{m}=0,
\end{aligned}
\]
and by direct subtraction and using \eqref{6.8} we have
\[
\begin{aligned}
&\hspace{-0.5cm}|f^n(t,x,v)-f^m(t,x,v)|\\
&\leq \int_0^t\Big(|D^{m-1}(s,X^{n-1}(s))-D^{n-1}(s,X^{n-1}(s))|\\
&\hspace{1cm}+|B^{m-1}(s,X^{n-1}(s))-B^{n-1}(s,X^{n-1}(s))|\Big)||\nabla_vf^m(s)||_\infty\,ds.
\end{aligned}
\]
By taking supremum, we get the desired result.
\end{proof}
%
%
\begin{lemma}
Consider the iteration functions \eqref{6.1}--\eqref{6.4}. Then we have
\[
\begin{aligned}
&\hspace{-0.5cm}||D^m(t)-D^n(t)||_\infty+||B^m(t)-B^n(t)||_\infty\\
&\leq C_T\Big(||f^m(t)-f^n(t)||_\infty+||\alpha^m(t)-\alpha^n(t)||_\infty
+||\beta^m(t)-\beta^n(t)||_\infty\Big).
\end{aligned}
\]
\end{lemma}
\begin{proof}
By \eqref{6.2}, we have
\[
D_1^m(t,x)-D_1^m(t,x)=\int_{-\infty}^x\int_{\bbr^2}f^m(t,y,v)-f^n(t,y,v)\,dv\,dy.
\]
We use \eqref{6.6} to have
\begin{equation}\label{6.12}
|D_1^m(t,x)-D_1^m(t,x)|\leq C_T||f^m(t)-f^n(t)||_\infty.
\end{equation}
By \eqref{6.4}, we have
\[
|D_2^m(t,x)-D_2^m(t,x)|\leq C_T(||f^m(t)-f^n(t)||_\infty+||\theta_2^m(t)-\theta_2^n(t)||_\infty),
\]
and
\[
|B^m(t,x)-B^m(t,x)|\leq C_T||\theta_B^m(t)-\theta_B^n(t)||_\infty.
\]
Note that
\begin{equation}\label{6.13}
||\theta_i^m(t)-\theta_i^n(t)||_\infty\leq ||\alpha^m(t)-\alpha^n(t)||_\infty+
||\beta^m(t)-\beta^n(t)||_\infty,\quad i\in\{2,B\},
\end{equation}
and therefore we get the desired result.
\end{proof}
%
%
\begin{lemma}
Consider the iteration functions \eqref{6.1}--\eqref{6.4}. Then we have
\[
\sum_{i=0}^2||k_i^m(t)-k_i^n(t)||_\infty\leq C_T||f^m(t)-f^n(t)||_\infty.
\]
\end{lemma}
\begin{proof}
Note that
\[
|j_i^m(t,x)-j_i^n(t,x)|\leq \int_{\bbr^2}|f^m(t,y,v)-f^n(t,y,v)|\,dv
\leq C_T||f^m(t)-f^n(t)||_\infty,\quad i\in\{0,1,2\},
\]
where we used \eqref{6.6}, and $\rho^n$ was denoted by $j_0$ for simplicity.
Together with \eqref{6.12}, we obtain the desired result.
\end{proof}
%
%
\begin{lemma}
Consider the iteration functions \eqref{6.1}--\eqref{6.4}. Then we have
\[
\begin{aligned}
&\hspace{-0.5cm}||\alpha^m(t)-\alpha^n(t)||_\infty+||\beta^m(t)-\beta^n(t)||_\infty\\
&\leq C_T\int_0^t||\alpha^m(\tau)-\alpha^n(\tau)||_\infty
+||\beta^m(\tau)-\beta^n(\tau)||_\infty+||f^m(\tau)-f^n(\tau)||_\infty\,d\tau.
\end{aligned}
\]
\end{lemma}
\begin{proof}
We use \eqref{6.3} for $m$ and $n$.
\[
\begin{aligned}
&\partial_t\alpha^{m}-(\cos\beta^{m})\partial_x\alpha^{m}=\cos\theta_2^{m}(k_0^{m}\sin\theta_B^{m}+k_1^{m}\sin\theta_2^{m}+k_2^{m}\cos\theta_2^{m}),\\
&\partial_t\alpha^{n}-(\cos\beta^{n})\partial_x\alpha^{n}=\cos\theta_2^{n}(k_0^{n}\sin\theta_B^{n}+k_1^{n}\sin\theta_2^{n}+k_2^{n}\cos\theta_2^{n}),
\end{aligned}
\]
After direct subtraction and applying the following characteristic curve:
\begin{equation}\label{6.14}
\frac{d}{d\tau}\xi^m(\tau)=-\cos\beta^m(\tau,\xi^m(\tau)),\quad \xi^m(t)=x,
\end{equation}
we make the same argument with Lemma 6.1. We use \eqref{6.13} and Lemma 6.3 to have
\[
\begin{aligned}
&\hspace{-0.5cm}|\alpha^m(t,x)-\alpha^n(t,x)|\\
&\leq C_T\int_0^t||\alpha^m(\tau)-\alpha^n(\tau)||_\infty
+||\beta^m(\tau)-\beta^n(\tau)||_\infty+||f^m(\tau)-f^n(\tau)||_\infty\,d\tau.
\end{aligned}
\]
By the same calculation, we obtain the estimate for $\beta^m-\beta^n$, and this completes the proof.
\end{proof}
%
%
\noindent{\bf Convergence to continuous functions.}
By Lemma 6.1, 6.2, and 6.4, it is proved that
the sequences $\{f^n\}$, $\{\alpha^n\}$, and $\{\beta^n\}$ converge to continuous functions.
For simplicity, we set
\[
\begin{aligned}
&f^{mn}(t):=||f^m(t)-f^n(t)||_\infty,\\
&\theta^{mn}(t):=||\alpha^m(t)-\alpha^n(t)||_\infty
+||\beta^m(t)-\beta^n(t)||_\infty.
\end{aligned}
\]
Then, the above results are written as follows:
\begin{align}
&f^{mn}(t)\leq C_T\int_0^t f^{(m-1)(n-1)}(s)+\theta^{(m-1)(n-1)}(s)\,ds,\label{6.15}\\
&\theta^{mn}(t)\leq C_T\int_0^t f^{mn}(s)+\theta^{mn}(s)\,ds.\label{6.16}
\end{align}
We apply Gr{\"o}nwall's inequality to \eqref{6.16} to have
\begin{equation}\label{6.17}
\theta^{mn}(t)\leq C_T\int_0^tf^{mn}(s)\,ds,
\end{equation}
and then apply it to \eqref{6.15} to have
\begin{eqnarray}
f^{mn}(t)&\leq& C_T\int_0^tf^{(m-1)(n-1)}(s)+\int_0^sf^{(m-1)(n-1)}(\tau)\,d\tau\,ds\cr
&=& C_T\int_0^tf^{(m-1)(n-1)}(s)\,ds+C_T\int_0^t\int^t_\tau f^{(m-1)(n-1)}(\tau)\,ds\,d\tau\cr
&\leq& C_T\int_0^tf^{(m-1)(n-1)}(s)\,ds.\label{6.18}
\end{eqnarray}
By iterating it, we obtain
\[
f^{mn}(t)\leq \frac{2}{k!}(C_T T)^k||f^{in}||_\infty\quad\mbox{for}\quad m,n>k,
\]
where the constant $C_T$ is the same one in the last inequality of \eqref{6.18}. This implies
that $\{f^n\}$ is a Cauchy sequence in $\mathcal{C}^0$ norm, so are $\{\alpha^n\}$ and $\{\beta^n\}$ due to \eqref{6.17}.
Consequently, the iteration functions converge to continuous functions $f$, $\alpha$, and $\beta$.
%
%
\subsubsection{Estimates for their derivatives} In this part, we estimate the derivatives of
$f^n$, $\alpha^n$, and $\beta^n$ to confirm that the solution is $\mathcal{C}^1$.
We introduce a notation $\varepsilon_{nm}$ for a small positive quantity which depends on $n$ and $m$
such that $\varepsilon_{nm}$ tends to zero as $n$ and $m$ go to infinity.
This quantity may depend on $T$ but not on the other variables, and its value will vary from line to line.
We also use a handy notation for $(x,v)\in\bbr^1\times\bbr^2$ as $z=(x,v)\in\bbr^3$.
Note that
\[
|z|=|(x,v)|\leq |x|+|v|\leq \sqrt{3}|z|,
\]
where $|\cdot|$ is the usual Euclidean norm on $\bbr^d$.
For simplicity, $Z^n$ will denote $(X^n,V^n)$ in some places, i.e.,
\[
Z^n(s;t,z)=(X^n(s;t,x,v),V^n(s;t,x,v))\in\bbr^3.
\]
In this part, we follow the arguments and notations of \cite{KRST05}.
%
%
\begin{lemma}
Consider the characteristic system \eqref{6.8}.
\begin{equation*}
\begin{aligned}
&\frac{d}{ds}X^n(s)=\hat{V}_1^n(s),\quad X^n(t)=x,\\
&\frac{d}{ds}V^n(s)=F^n(s,X^n(s),V^n(s)),\quad V^n(t)=v,
\end{aligned}
\end{equation*}
where $F^n(t,z)=D^n(t,x)+(\hat{v}_2,-\hat{v}_1)B^n(t,x)$.
For any $n\neq m$, we have the following estimates:
\begin{align*}
(i)&\quad |X^n(s;t,z)-X^m(s;t,z)|+|V^n(s;t,z)-V^m(s;t,z)|\leq \varepsilon_{nm}.\\
(ii)&\quad |\nabla_zX^n(s;t,z)|+|\nabla_zV^n(s;t,z)|\leq C_T.\\
(iii)&\quad |\nabla_zX^n(s;t,z)-\nabla_zX^m(s;t,z)|+|\nabla_zV^n(s;t,z)-\nabla_zV^m(s;t,z)|\\
&\hspace{1cm}\leq \varepsilon_{nm}
+C_T\int_s^t |(\partial_xF^n)(\tau,X^n(\tau))-(\partial_xF^n)(\tau,X^m(\tau))|
+||\partial_xF^n(\tau)-\partial_xF^m(\tau)||_\infty\,d\tau,
\end{align*}
where $\varepsilon_{mn}$ is a small positive quantity which tends to zero as $n,m\rightarrow\infty$
and depends on $T$ but not on $x$ and $v$.
\end{lemma}
\begin{proof}
$(i)$ We consider the characteristic systems for $n$ and $m$ together with the same initial data,
$Z^n(t)=Z^m(t)=z$.
By integrating the systems from $s$ to $t$, we have
\[
|X^n(s;t,z)-X^m(s;t,z)|\leq \int_s^t|V^n(\tau;t,z)-V^m(\tau;t,z)|\,d\tau,
\]
and
\begin{align*}
&\hspace{-0.5cm}|V^n(s;t,z)-V^m(s;t,z)|\cr
&\leq \int_s^t|F^n(\tau,X^n(\tau))-F^n(\tau,X^m(\tau))|
+|F^n(\tau,X^m(\tau))-F^m(\tau,X^m(\tau))|\,d\tau\cr
&\leq C_T\int_s^t|X^n(\tau;t,z)-X^m(\tau;t,z)|\,d\tau+\varepsilon_{nm},
\end{align*}
where we used the fact that $\nabla_xF^n$ is bounded and $\{F^n\}$ is a Cauchy sequence.
We combine the above two inequality to have the following estimate:
\[
|X^n(s;t,z)-X^m(s;t,z)|\leq \varepsilon_{nm}+C_T\int_s^t|X^n(\tau;t,z)-X^m(\tau;t,z)|\,d\tau,
\]
where the constant $\varepsilon_{nm}$ is another small constant depending on $T$,
but it still satisfies the property in the statement of the lemma. By applying the Gr{\"o}nwall inequality again,
we obtain the desired result for $X^n$,
\[
|X^n(s;t,z)-X^m(s;t,z)|\leq \varepsilon_{nm}e^{C_T(t-s)}\leq \varepsilon_{nm},
\]
and for $V^n$,
\[
|V^n(s;t,z)-V^m(s;t,z)|\leq \varepsilon_{nm},
\]
and this proves the first result.\bigskip

\noindent $(ii)$ By direct differentiation with respect to $x$, we obtain
\begin{align}
\partial_x\dot{X}^n(s;t,z)&=\nabla_{V^n}\hat{V}_1^n\cdot\partial_xV^n(s;t,z),\label{6.19}\\
\partial_x\dot{V}^n(s;t,z)&=(\partial_xD^n)(s,X^n(s;t,z))\partial_xX^n(s;t,z)\label{6.20}\\
&\hspace{0.5cm}+(\nabla_{V^n}\hat{V}^n_2,-\nabla_{V^n}\hat{V}^n_1)\partial_xV^n(s;t,z) B^n(s,X^n(s;t,z))\cr
&\hspace{0.5cm}+(\hat{V}^n_2,-\hat{V}^n_1)(\partial_xB^n)(s,X^n(s;t,z))\partial_xX^n(s;t,z),\nonumber
\end{align}
where the dots denote $s$ derivatives.
Since $\partial_xX^n(t)=1$ and $\partial_xV^n_i(t)=0$, $i=1,2$,
by integrating from $s$ to $t$ we obtain
\begin{align*}
&|\partial_xX^n(s)|\leq 1+C\int_s^t|\partial_xV^n(\sigma)|\,d\sigma,\\
&|\partial_xV^n(s)|\leq C_T\int_s^t|\partial_xX^n(\sigma)|+|\partial_xV^n(\sigma)|\,d\sigma,
\end{align*}
and then Gr{\"o}nwall's inequality gives the desired result after we estimate $v$ derivative quantities by
using the same arguments.\bigskip

\noindent $(iii)$ To prove the third estimate, we use \eqref{6.19}--\eqref{6.20} for $n$ and $m$.
Note that the map $v\mapsto\hat{v}$ is a smooth function with a bounded $\mathcal{C}^1$ norm.
\begin{align*}
&\hspace{-0.5cm}|\partial_xX^n(t)-\partial_xX^m(t)|\\
&\leq\int_s^t|\nabla_{V^n}\hat{V}^n_1-\nabla_{V^m}\hat{V}^m_1||\partial_xV^n(\tau)|
+|\nabla_{V^m}\hat{V}^m_1||\partial_xV^n(\tau)-\partial_xV^m(\tau)|\,d\tau\\
&\leq\varepsilon_{nm}+C\int_s^t|\partial_xV^n(\tau)-\partial_xV^m(\tau)|\,d\tau,
\end{align*}
where we used $(i)$ and $(ii)$. By the same way, we have
\begin{align*}
&\hspace{-0.5cm}|\partial_xV^n(t)-\partial_xV^m(t)|\\
&\leq\varepsilon_{nm}+ C_T\int_s^t|(\partial_xD^n)(\tau,X^n(\tau))-(\partial_xD^m)(\tau,X^m(\tau))|
+|\partial_xX^n(\tau)-\partial_xX^m(\tau)|\,d\tau\\
&\hspace{0.5cm}+C_T\int_s^t|\partial_xV^n(\tau)-\partial_xV^m(\tau)|
+|B^n(\tau,X^n(\tau))-B^m(\tau,X^m(\tau))|\,d\tau\\
&\hspace{0.5cm}+C_T\int_s^t|(\partial_xB^n)(\tau,X^n(\tau))-(\partial_xB^m)(\tau,X^m(\tau))|\,d\tau,
\end{align*}
where we used $(i)$ and $(ii)$ together with the fact that the quantities in \eqref{6.11} are bounded.
In the above inequality, we can see that
\begin{align*}
&\hspace{-0.5cm}|(\partial_xD^n)(\tau,X^n(\tau))-(\partial_xD^m)(\tau,X^m(\tau))|\\
&\leq |(\partial_xD^n)(\tau,X^n(\tau))-(\partial_xD^n)(\tau,X^m(\tau))|
+|(\partial_xD^n)(\tau,X^m(\tau))-(\partial_xD^m)(\tau,X^m(\tau))|\\
&\leq |(\partial_xD^n)(\tau,X^n(\tau))-(\partial_xD^n)(\tau,X^m(\tau))|
+||\partial_xD^n(\tau)-\partial_xD^m(\tau)||_\infty.
\end{align*}
By the same arguments, we have the similar estimate for $\partial_xB^n-\partial_xB^m$.
\begin{align*}
&\hspace{-0.5cm}|(\partial_xB^n)(\tau,X^n(\tau))-(\partial_xB^m)(\tau,X^m(\tau))|\\
&\leq |(\partial_xB^n)(\tau,X^n(\tau))-(\partial_xB^n)(\tau,X^m(\tau))|
+||\partial_xB^n(\tau)-\partial_xB^m(\tau)||_\infty.
\end{align*}
For $B^n-B^m$ term, we have
\begin{align*}
&\hspace{-0.5cm}|B^n(\tau,X^n(\tau))-B^m(\tau,X^m(\tau))|\\
&\leq|B^n(\tau,X^n(\tau))-B^n(\tau,X^m(\tau))|+|B^n(\tau,X^m(\tau))-B^m(\tau,X^m(\tau))|\\
&\leq \varepsilon_{nm},
\end{align*}
where we used $(i)$, boundedness of $\nabla_xB^n$, and the Cauchy property of $\{B^n\}$.
We combine the above results to obtain
\begin{align*}
&\hspace{-0.5cm}|\partial_xV^n(t)-\partial_xV^m(t)|\\
&\leq\varepsilon_{nm}+ C_T\int_s^t|\partial_xX^n(\tau)-\partial_xX^m(\tau)|\,d\tau
+|\partial_xV^n(\tau)-\partial_xV^m(\tau)|\,d\tau\\
&\hspace{0.5cm}+C_T\int_s^t |(\partial_xF^n)(\tau,X^n(\tau))-(\partial_xF^n)(\tau,X^m(\tau))|
+||\partial_xF^n(\tau)-\partial_xF^m(\tau)||_\infty\,d\tau.
\end{align*}
By applying Gr{\"o}nwall's inequality, we obtain the estimate for $x$ derivatives.
The estimate for $v$ derivatives is verified by the similar calculations, and we obtain the desired result.
\end{proof}
%
%
\begin{lemma}
Consider the following two characteristic equations:
\begin{align}
&\frac{d}{d\tau}\xi^n(\tau)=-\cos\beta^n(\tau,\xi^n(\tau)),\quad \xi^n(t)=x,\label{6.21}\\
&\frac{d}{d\tau}\eta^n(\tau)=\cos\alpha^n(\tau,\eta^n(\tau)),\quad \eta^n(t)=x.\label{6.22}
\end{align}
For any $n\neq m$, we have
\begin{align*}
(i)&\quad |\xi^n(\tau;t,x)-\xi^m(\tau;t,x)|+|\eta^n(\tau;t,x)-\eta^m(\tau;t,x)|\leq \varepsilon_{nm}.\\
(ii)&\quad |\partial_x\xi^n(\tau;t,x)|+|\partial_x\eta^n(\tau;t,x)|\leq C_T.\\
(iii)&\quad |\partial_x\xi^n(\tau;t,x)-\partial_x\xi^m(\tau;t,x)|\\
&\hspace{1cm}\leq \varepsilon_{nm}+C_T\int_\tau^t
|(\partial_x\beta^n)(s,\xi^n(s))-(\partial_x\beta^n)(s,\xi^m(s))|
+||\partial_x\beta^n(s)-\partial_x\beta^m(s)||_\infty\,ds.\\
(iv)&\quad |\partial_x\eta^n(\tau;t,x)-\partial_x\eta^m(\tau;t,x)|\\
&\hspace{1cm}\leq \varepsilon_{nm}+C_T\int_\tau^t
|(\partial_x\alpha^n)(s,\eta^n(s))-(\partial_x\alpha^n)(s,\eta^m(s))|
+||\partial_x\alpha^n(s)-\partial_x\alpha^m(s)||_\infty\,ds,
\end{align*}
where $\varepsilon_{mn}$ is a small positive quantity which tends to zero as $n,m\rightarrow\infty$
and depends on $T$ but not on $x$ and $v$.
\end{lemma}
\begin{proof}
$(i)$ The proof is almost the same with that of the previous lemma. We first consider the characteristic equation \eqref{6.21}
for $n$ and $m$ together with the same initial data $\xi^n(t)=\xi^m(t)=x$.
By integrating the equation from $\tau$ to $t$, we have
\begin{align*}
&\hspace{-0.5cm}|\xi^n(\tau;t,x)-\xi^m(\tau;t,x)|\cr
&\leq\int_\tau^t|\cos\beta^n(s,\xi^n(s))-\cos\beta^m(s,\xi^n(s))|
+|\cos\beta^m(s,\xi^n(s))-\cos\beta^m(s,\xi^m(s))|\,ds\cr
&\leq \varepsilon_{nm}+C_T\int_\tau^t |\xi^n(s;t,x)-\xi^m(s;t,x)|\,ds,
\end{align*}
where we used the fact that $\{\beta^n\}$ is a Cauchy sequence and $\partial_x\beta^m$ is bounded.
By applying the Gr{\"o}nwall inequality, we obtain the desired result for $\xi^n$.
\[
|\xi^n(\tau;t,x)-\xi^m(\tau;t,x)|\leq\varepsilon_{nm}.
\]
We apply the same argument to \eqref{6.22} to obtain the second estimate for $\eta^n$.
\[
|\eta^n(\tau;t,x)-\eta^m(\tau;t,x)|\leq\varepsilon_{nm}.
\]\bigskip

\noindent$(ii)$ By direct differentiation with respect to $x$, we obtain
\begin{equation}\label{6.23}
\partial_x\dot{\xi}^n(\tau;t,x)=\sin\beta^n(\tau,\xi^n(\tau))(\partial_x\beta^n)(\tau,\xi^n(\tau))\partial_x\xi^n(\tau),
\end{equation}
which gives
\[
|\partial_x\xi^n(\tau)|
\leq 1+C_T\int_\tau^t|\partial_x\xi^n(s)|\,ds,
\]
and we obtain the desired result by Gr{\"o}nwall's inequality. The estimate for $\eta^n$ is given by the same
way, and we skip it.\bigskip

\noindent$(iii)$ The proof of the third estimate is almost the same with that of Lemma 6.5 $(iii)$: we use \eqref{6.23}
instead of \eqref{6.19}--\eqref{6.20}, consider the equation \eqref{6.23} for $n$ and $m$, and
use the result $(i)$, boundedness of $\partial_x\beta^n$, and the Cauchy property of $\{\beta^n\}$,
and then we obtain the desired result.\bigskip

\noindent$(iv)$ The last estimate is similarly verified as in $(iii)$, and we skip the proof.
\end{proof}
%
%
\begin{lemma}
Consider the iteration functions \eqref{6.1}--\eqref{6.4}. Then we have
\[
||\partial_x F^m(t)-\partial_xF^n(t)||_\infty\leq\varepsilon_{nm}+C_T
\Big(||\partial_x\alpha^m(t)-\partial_x\alpha^n(t)||_\infty
+||\partial_x\beta^m(t)-\partial_x\beta^n(t)||_\infty\Big),
\]
where $F^n(t,x,v)=D^n(t,x)+(\hat{v}_2,-\hat{v}_1)B^n(t,x)$ and
$\varepsilon_{mn}$ is a small positive quantity which tends to zero as $n,m\rightarrow\infty$
and depends on $T$ but not on $x$ and $v$.
\end{lemma}
\begin{proof}
We note that the first component of $\partial_xD^m-\partial_xD^n$ is written as follows:
\[
\partial_xD^m_1-\partial_xD^n_1=\rho^m-\rho^n=\int f^m-f^n\,dv.
\]
Since $\{f^n\}$ is Cauchy and the momentum supports of $f^n$ are bounded by $P(T)$ uniformly on $n$,
we have
\[
|\partial_xD^m_1-\partial_xD^n_1|\leq\varepsilon_{nm}.
\]
We use \eqref{6.4} to estimate the second component of $\partial_xD^m-\partial_xD^n$, i.e.,
\[
\partial_xD^n_2=\frac{D_1^n\partial_xD^n_1}{\sqrt{1+(D_1^n)^2}}\tan\left(\frac{\alpha^n+\beta^n}{2}\right)
+\sqrt{1+(D_1^n)^2}\frac{1}{\cos^2\left(\frac{\alpha^n+\beta^n}{2}\right)}\frac{\partial_x\alpha^n+\partial_x\beta^n}{2}.
\]
By direct calculations, we have
\begin{align*}
&\hspace{-0.5cm}|\partial_xD^m_2-\partial_xD^n_2|\cr
&\leq C_T\sum_{i=0}^1\Big(|\partial_x^iD_1^m-\partial_x^iD_1^n|
+|\partial_x^i\alpha^m-\partial_x^i\alpha^n|
+|\partial_x^i\beta^m-\partial_x^i\beta^n|\Big)\cr
&\leq \varepsilon_{nm}+C_T\Big(|\partial_x\alpha^m-\partial_x\alpha^n|
+|\partial_x\beta^m-\partial_x\beta^n|\Big),
\end{align*}
and the same argument gives the estimate for $\partial_xB^m-\partial_xB^n$.
\[
|\partial_xB^m-\partial_xB^n|\leq \varepsilon_{nm}+C_T\Big(|\partial_x\alpha^m-\partial_x\alpha^n|
+|\partial_x\beta^m-\partial_x\beta^n|\Big).
\]
This completes the proof of the lemma.
\end{proof}
%
%
\begin{lemma}
Consider the characteristic system \eqref{6.8}.
\begin{equation*}
\begin{aligned}
&\frac{d}{ds}X^n(s)=\hat{V}_1^n(s),\quad X^n(t)=x,\\
&\frac{d}{ds}V^n(s)=F^n(s,X^n(s),V^n(s)),\quad V^n(t)=v,
\end{aligned}
\end{equation*}
where $F^n(t,z)=D^n(t,x)+(\hat{v}_2,-\hat{v}_1)B^n(t,x)$.
For any $z\neq z'$, we have the following estimates:
\begin{align*}
(i)&\quad|X^n(s;t,z)-X^n(s;t,z')|
+|V^n(s;t,z)-V^n(s;t,z')|
\leq C_T|z-z'|.\\
(ii)&\quad |\nabla_zX^n(s;t,z)-\nabla_zX^n(s;t,z')|+|\nabla_zV^n(s;t,z)-\nabla_zV^n(s;t,z')|\\
&\hspace{1cm}\leq C_T|z-z'|
+C_T\int_s^t|(\partial_xD^n)(\sigma,X^n(\sigma;t,z))-(\partial_xD^n)(\sigma,X^n(\sigma;t,z'))|\,d\sigma\\
&\hspace{1.5cm}+C_T\int_s^t|(\partial_xB^n)(\sigma,X^n(\sigma;t,z))-(\partial_xB^n)(\sigma,X^n(\sigma;t,z'))|\,d\sigma,
\end{align*}
where the constant $C_T$ does not depend on $n$.
\end{lemma}
\begin{proof}
$(i)$ We apply the mean value theorem to $Z^n(s;t,z)-Z^n(s;t,z')$
and use Lemma 6.5 $(ii)$, and then the first estimate is obtained.\bigskip

\noindent $(ii)$ Let $z\neq z'$, and consider two characteristic curves for $(\sigma;t,z)$ and $(\sigma;t,z')$.
By using \eqref{6.19}--\eqref{6.20}, we have
\[
|\partial_xX^n(s;t,z)-\partial_xX^n(s;t,z')|\leq
C_T|z-z'|+C_T\int_s^t|\partial_xV^n(\sigma;t,z)-\partial_xV^n(\sigma;t,z')|\,d\sigma,
\]
where we used the result $(i)$ and Lemma 6.5 $(ii)$. In the same way, we obtain
\begin{align*}
&\hspace{-0.5cm}|\partial_xV^n(s;t,z)-\partial_xV^n(s;t,z')|\\
&\leq C_T\int_s^t|(\partial_xD^n)(\sigma,X^n(\sigma;t,z))-(\partial_xD^n)(\sigma,X^n(\sigma;t,z'))|\,d\sigma\\
&\hspace{0.5cm}+C_T\int_s^t|\partial_xX^n(s;t,z)-\partial_xX^n(s;t,z')|\,d\sigma\\
&\hspace{0.5cm}+C_T\int_s^t|z-z'|+|\partial_xV^n(s;t,z)-\partial_xV^n(s;t,z')|\,d\sigma\\
&\hspace{0.5cm}+C_T\int_s^t|(\partial_xB^n)(\sigma,X^n(\sigma;t,z))-(\partial_xB^n)(\sigma,X^n(\sigma;t,z'))|\,d\sigma,
\end{align*}
where we used $(i)$, Lemma 6.5 $(ii)$, and the fact that $\partial_xD^n$, $B^n$, and $\partial_xB^n$ are uniformly bounded.
Gr{\"o}nwall's inequality again gives the desired result after we estimate $v$ derivative quantities
by the same arguments.
\end{proof}
%
%
\begin{remark}
We can see that $\partial_xD^n(t,x)-\partial_xD^n(t,x')$
and $\partial_xB^n(t,x)-\partial_xB^n(t,x')$ are estimated by the transformed variables
$\partial_x\alpha^n$ and $\partial_x\beta^n$.
If we use \eqref{6.2} and \eqref{6.4}, then the following inequality is easily obtained:
\begin{equation}\label{6.24}
\begin{aligned}
&\hspace{-0.5cm}|\partial_xD^n(t,x)-\partial_xD^n(t,x')|
+|\partial_xB^n(t,x)-\partial_xB^n(t,x')|\\
&\leq C_T(|x-x'|+|\partial_x\alpha^n(t,x)-\partial_x\alpha^n(t,x')|
+|\partial_x\beta^n(t,x)-\partial_x\beta^n(t,x')|).
\end{aligned}
\end{equation}
Hence, the second result of Lemma 6.8 implies the following:
\begin{align*}
&\hspace{-0.5cm}|\nabla_zX^n(s;t,z)-\nabla_zX^n(s;t,z')|+|\nabla_zV^n(s;t,z)-\nabla_zV^n(s;t,z')|\\
&\leq C_T|z-z'|
+C_T\int_s^t|(\partial_x\alpha^n)(\sigma,X^n(\sigma;t,z))-(\partial_x\alpha^n)(\sigma,X^n(\sigma;t,z'))|\,d\sigma\\
&\hspace{0.5cm}+C_T\int_s^t|(\partial_x\beta^n)(\sigma,X^n(\sigma;t,z))-(\partial_x\beta^n)(\sigma,X^n(\sigma;t,z'))|\,d\sigma.
\end{align*}
\end{remark}
%
%
\begin{lemma}
Consider the characteristic equations \eqref{6.21} and \eqref{6.22}.
\begin{align*}
&\frac{d}{d\tau}\xi^n(\tau)=-\cos\beta^n(\tau,\xi^n(\tau)),\quad \xi^n(t)=x,\\
&\frac{d}{d\tau}\eta^n(\tau)=\cos\alpha^n(\tau,\eta^n(\tau)),\quad \eta^n(t)=x.
\end{align*}
For any $x\neq x'$, we have the following estimates:
\begin{align*}
(i)&\quad |\xi^n(s;t,x)-\xi^n(s;t,x')|+|\eta^n(s;t,x)-\eta^n(s;t,x')|\leq C_T|x-x'|.\\
(ii)&\quad |\partial_x\xi^n(s;t,x)-\partial_x\xi^n(s;t,x')|\\
&\hspace{1cm}\leq C_T|x-x'|
+C_T\int_s^t|\partial_x\beta^n(\tau,\xi^n(\tau;t,x))-\partial_x\beta^n(\tau,\xi^n(\tau;t,x'))|\,d\tau.\\
(iii)&\quad |\partial_x\eta^n(s;t,x)-\partial_x\eta^n(s;t,x')|\\
&\hspace{1cm}\leq C_T|x-x'|
+C_T\int_s^t|\partial_x\alpha^n(\tau,\eta^n(\tau;t,x))-\partial_x\alpha^n(\tau,\eta^n(\tau;t,x'))|\,d\tau,
\end{align*}
where the constant $C_T$ does not depend on $n$.
\end{lemma}
\begin{proof}
The argument of the proof is straightforward and almost the same with that of Lemma 6.8, hence we skip the proof.
\end{proof}
%
%
\begin{lemma}
Consider the iteration functions \eqref{6.1}--\eqref{6.4}.
For any $z\neq z'$, we have
\begin{align*}
&\hspace{-0.5cm}|\nabla_zf^{n+1}(t,z)-\nabla_zf^{n+1}(t,z')|\\
&\leq C_T|(\nabla_zf^{in})(Z^n(0;t,z))
-(\nabla_zf^{in})(Z^n(0;t,z'))|\\
&\hspace{0.5cm}+C_T|z-z'|
+C_T\int_0^t|(\partial_x\alpha^n)(s,X^n(s;t,z))-(\partial_x\alpha^n)(s,X^n(s;t,z'))|\,ds\\
&\hspace{0.5cm}+C_T\int_0^t|(\partial_x\beta^n)(s,X^n(s;t,z))-(\partial_x\beta^n)(s,X^n(s;t,z'))|\,ds.
\end{align*}
\end{lemma}
\begin{proof}
We only consider $x$ derivative of $f^{n+1}$ since the calculation for $v$ derivatives is almost same.
Since $f^{n+1}(t,z)=f^{in}(Z^n(0;t,z))$, we have
\begin{align*}
&\hspace{-0.5cm}|\partial_xf^{n+1}(t,z)-\partial_xf^{n+1}(t,z')|\\
&\leq |(\nabla_zf^{in})(Z^n(0;t,z))-(\nabla_zf^{in})(Z^n(0;t,z'))||\partial_xZ^n(0;t,z)|\\
&\hspace{0.5cm}+|(\nabla_zf^{in})(Z^n(0;t,z'))||\partial_xZ^n(0;t,z)-\partial_xZ^n(0;t,z')|\\
&\leq C_T|(\nabla_zf^{in})(Z^n(0;t,z))-(\nabla_zf^{in})(Z^n(0;t,z'))|\\
&\hspace{0.5cm}+C_T|z-z'|
+C_T\int_0^t|(\partial_x\alpha^n)(s,X^n(s;t,z))-(\partial_x\alpha^n)(s,X^n(s;t,z'))|\,ds\\
&\hspace{0.5cm}+C_T\int_0^t|(\partial_x\beta^n)(s,X^n(s;t,z))-(\partial_x\beta^n)(s,X^n(s;t,z'))|\,ds,
\end{align*}
where we used Lemma 6.5 $(ii)$, boundedness of $\mathcal{C}^1$ norm of initial data, and Remark 6.1.
This completes the proof.
\end{proof}
%
%
\begin{lemma}
Consider the iteration functions \eqref{6.1}--\eqref{6.4}. For any $x\neq x'$, we have
the following estimates:
\begin{align*}
(i)&\quad|\partial_x\alpha^{n}(t,x)-\partial_x\alpha^{n}(t,x')|\\
&\hspace{1cm}\leq C_T|(\partial_x\alpha^{in})(\xi^n(0;t,x))
-(\partial_x\alpha^{in})(\xi^n(0;t,x'))|\\
&\hspace{1.5cm}+C_T|x-x'|
+C_T\int_0^t|(\partial_x\alpha^n)(\tau,\xi^n(\tau;t,x))-(\partial_x\alpha^n)(\tau,\xi^n(\tau;t,x'))|\,d\tau\\
&\hspace{1.5cm}+C_T\int_0^t|(\partial_x\beta^n)(\tau,\xi^n(\tau;t,x))-(\partial_x\beta^n)(\tau,\xi^n(\tau;t,x'))|\,d\tau\\
&\hspace{1.5cm}+C_T\int_0^t\int|(\partial_xf^n)(\tau,\xi^n(\tau;t,x),v)
-(\partial_xf^n)(\tau,\xi^n(\tau;t,x'),v)|\,dv\,d\tau.\\
(ii)&\quad|\partial_x\beta^{n}(t,x)-\partial_x\beta^{n}(t,x')|\\
&\hspace{1cm}\leq C_T|(\partial_x\beta^{in})(\eta^n(0;t,x))
-(\partial_x\beta^{in})(\eta^n(0;t,x'))|\\
&\hspace{1.5cm}+C_T|x-x'|
+C_T\int_0^t|(\partial_x\alpha^n)(\tau,\eta^n(\tau;t,x))-(\partial_x\alpha^n)(\tau,\eta^n(\tau;t,x'))|\,d\tau\\
&\hspace{1.5cm}+C_T\int_0^t|(\partial_x\beta^n)(\tau,\eta^n(\tau;t,x))-(\partial_x\beta^n)(\tau,\eta^n(\tau;t,x'))|\,d\tau\\
&\hspace{1.5cm}+C_T\int_0^t\int|(\partial_xf^n)(\tau,\eta^n(\tau;t,x),v)
-(\partial_xf^n)(\tau,\eta^n(\tau;t,x'),v)|\,dv\,d\tau.
\end{align*}
\end{lemma}
\begin{proof}
We first consider $\partial_xj^n_i(t,x)-\partial_xj^n_i(t,x')$, $i=0,1,2$, where $j^n_0=\rho^n$,
which can be easily estimated by \eqref{6.2} as follows:
\[
|\partial_xj^n_i(t,x)-\partial_xj^n_i(t,x')|
\leq C|x-x'|+\int |\partial_xf^n(t,x,v)-\partial_xf^n(t,x',v)|\,dv
\]
for any $i=0,1,2$.
Then, we use \eqref{5.14}--\eqref{5.16} and the fact that the quantities in \eqref{6.11} are uniformly
bounded by $C_T$ to obtain
\[
|\partial_xk_i(t,x)-\partial_xk_i(t,x')|\leq
C_T|x-x'|+C_T\int |\partial_xf^n(t,x,v)-\partial_xf^n(t,x',v)|\,dv
\]
for any $i=0,1,2$.
The proof is now straightforward. We use \eqref{5.6}, take $x$ derivative on it,
and estimate it for $x$ and $x'$ by using Lemma 6.9
and the above inequality together with the fact that the quantities in \eqref{6.11}
are uniformly bounded by $C_T$, and then we obtain $(i)$. The second estimate
is obtained by the same argument, and this completes the proof.
\end{proof}
%
%
\begin{lemma}
Consider the iteration functions \eqref{6.1}--\eqref{6.4}.
The sets of derivatives of the iteration functions $\{\nabla_zf^n\}$, $\{\partial_x\alpha^n\}$,
and $\{\partial_x\beta^n\}$ are equicontinuous.
\end{lemma}
\begin{proof}
We use Lemma 6.8--6.11.
We first define the following quantities:
\begin{align*}
&\varepsilon_n(t,\delta):=\sup\{|\nabla_zf^n(t,z)-\nabla_zf^n(t,z')|: |z-z'|\leq\delta\},\\
&\theta_n(t,\delta):=\sup\{|\partial_x\alpha^n(t,x)-\partial_x\alpha^n(t,x')|
+|\partial_x\beta^n(t,x)-\partial_x\beta^n(t,x')|: |x-x'|\leq \delta\}.
\end{align*}
Note that for any positive integer $N$ we have
\[
\varepsilon_n(t,\delta)\leq\varepsilon_n(t,N\delta)
\leq N\varepsilon_n(t,\delta).
\]
On the other hand, we can see that for any $n$
\[
\varepsilon_n(0,\delta)=\varepsilon_0(0,\delta)=:\varepsilon_0(\delta)\quad\mbox{and}\quad
\theta_n(0,\delta)=\theta_0(0,\delta)=:\theta_0(\delta),
\]
which are determined by given $\mathcal{C}^1$ initial data.
By Lemma 6.8 and 6.9, we can choose a positive integer $M$ such that
\[
|Z^n(s;t,z)-Z^n(s;t,z')|
+|\xi^n(s;t,x)-\xi^n(s;t,x')|+|\eta^n(s;t,x)-\eta^n(s;t,x')|\leq M\delta
\]
for any $\delta>0$, $|z-z'|\leq\delta$, $n\geq 0$, and $0\leq s\leq t\leq T$.
We now fix a $\delta>0$ and
use $\varepsilon_n(t,\delta)$ and $\theta_n(t,\delta)$ to rewrite Lemma 6.10 and 6.11 as follows:
\begin{align*}
\varepsilon_{n+1}(t,\delta)&\leq C_T\varepsilon_0(M\delta)+C_T\delta
+C_T\int_0^t\theta_n(s,M\delta)\,ds\\
&\leq C_TM\varepsilon_0(\delta)+C_T\delta
+C_TM\int_0^t\theta_n(s,\delta)\,ds\\
&\leq C_T\left(\delta+\varepsilon_0(\delta)+\int_0^t\theta_n(s,\delta)\,ds\right),
\end{align*}
and similarly
\begin{align*}
\theta_n(t,\delta)&\leq 2C_T\theta_0(M\delta)+2C_T\delta+2C_T\int_0^t\theta_n(\tau,M\delta)\,d\tau
+2C_T\int_0^t\int_{|v|\leq P(\tau)}\varepsilon_n(\tau,M\delta)\,dv\,d\tau\\
&\leq C_T\left(\delta+\theta_0(\delta)+\int_0^t\theta_n(\tau,\delta)\,d\tau
+\int_0^t\varepsilon_n(\tau,\delta)\,d\tau\right).
\end{align*}
We apply Gr{\"o}nwall's inequality to the last inequality and then iterate the above two inequalities
to conclude that for large $n$ we have
\[
\varepsilon_n(t,\delta)+\theta_n(t,\delta)\leq C_T(\delta+\varepsilon_0(\delta)+\theta_0(\delta))
\]
for any $0\leq t\leq T$, and this implies that $\{\nabla_zf^n\}$, $\{\partial_x\alpha^n\}$,
and $\{\partial_x\beta^n\}$ are equicontinuous.
\end{proof}
%
%
\begin{lemma}
Consider the iteration functions \eqref{6.1}--\eqref{6.4}. Then we have
\begin{align*}
&\hspace{-0.5cm} ||\nabla_zf^{n+1}(t)-\nabla_zf^{m+1}(t)||_\infty\\
&\leq \varepsilon_{nm}+C_T\int_0^t||\partial_x\alpha^n(s)-\partial_x\alpha^m(s)||_\infty
+||\partial_x\beta^n(s)-\partial_x\beta^m(s)||_\infty\,ds,
\end{align*}
where $\varepsilon_{nm}$ is a small positive quantity which tends to zero as $n,m\rightarrow\infty$
and depends on $T$ but not on $x$ and $v$.
\end{lemma}
\begin{proof}
Since $f^{n+1}(t,z)=f^{in}(Z^n(0;t,z))$, we have
\begin{align*}
&\hspace{-0.5cm}|\nabla_zf^{n+1}(t,z)-\nabla_zf^{m+1}(t,z)|\\
&\leq |(\nabla_zf^{in})(Z^n(0;t,z))-(\nabla_zf^{in})(Z^m(0;t,z))||\nabla_xZ^n(0;t,z)|\\
&\hspace{0.5cm}+|(\nabla_zf^{in})(Z^m(0;t,z))||\nabla_xZ^n(0;t,z)-\nabla_xZ^m(0;t,z)|.
\end{align*}
We use Lemma 6.5 $(ii)$, $(iii)$, and then \eqref{6.24} and Lemma 6.7 to obtain
\begin{align*}
&\hspace{-0.5cm}|\nabla_zf^{n+1}(t,z)-\nabla_zf^{m+1}(t,z)|\\
&\leq \varepsilon_{nm}+C_T|(\nabla_zf^{in})(Z^n(0;t,z))-(\nabla_zf^{in})(Z^m(0;t,z))|\\
&\hspace{0.5cm}+C_T\int_0^t|(\partial_x\alpha^n)(\tau,X^n(\tau))-(\partial_x\alpha^n)(\tau,X^m(\tau))|\,d\tau\\
&\hspace{0.5cm}+C_T\int_0^t|(\partial_x\beta^n)(\tau,X^n(\tau))-(\partial_x\beta^n)(\tau,X^m(\tau))|\,d\tau\\
&\hspace{0.5cm}+C_T\int_0^t||\partial_x\alpha^n(\tau)-\partial_x\alpha^m(\tau)||_\infty
+||\partial_x\beta^n(\tau)-\partial_x\beta^m(\tau)||_\infty\,d\tau,
\end{align*}
where we also used Lemma 6.5 $(i)$.
We now use the equicontinuity of $\{\nabla_zf^n\}$, $\{\partial_x\alpha^n\}$, and $\{\partial_x\beta^n\}$
in Lemma 6.12. Since $Z^n-Z^m$ is $\varepsilon_{nm}$, the second, third, and fourth quantities in the RHS
of the above inequality are also $\varepsilon_{nm}$ which is independent of $\partial_x\alpha^n$ and $\partial_x\beta^n$.
This completes the proof of the lemma.
\end{proof}
%
%
\begin{lemma}
Consider the iteration functions \eqref{6.1}--\eqref{6.4}. Then we have
\begin{align*}
&\hspace{-0.5cm} ||\partial_x\alpha^n(t)-\partial_x\alpha^m(t)||_\infty
+||\partial_x\beta^n(t)-\partial_x\beta^m(t)||_\infty\\
&\leq \varepsilon_{nm}+C_T\int_0^t||\partial_x\alpha^n(s)-\partial_x\alpha^m(s)||_\infty
+||\partial_x\beta^n(s)-\partial_x\beta^m(s)||_\infty\,ds\\
&\hspace{0.5cm}+C_T\int_0^t||\nabla_zf^n(s)-\nabla_zf^m(s)||_\infty\,ds,
\end{align*}
where $\varepsilon_{nm}$ is a small positive quantity which tends to zero as $n,m\rightarrow\infty$
and depends on $T$ but not on $x$ and $v$.
\end{lemma}
\begin{proof}
As in the proof of Lemma 6.13, we use the equicontinuity of $\{\partial_x\alpha^n\}$ and $\{\partial_x\beta^n\}$.
We take $x$ derivative on \eqref{5.6}, and then apply Lemma 6.6 to have
\begin{align*}
&\hspace{-0.5cm}|\partial_x\alpha^n(t,x)-\partial_x\alpha^m(t,x)|\\
&\leq\varepsilon_{nm}+C_T\int_0^t||\partial_x\alpha^n(s)-\partial_x\alpha^m(s)||_\infty
+||\partial_x\beta^n(s)-\partial_x\beta^m(s)||_\infty\,ds\\
&\hspace{0.5cm}+C_T\sum_{i=0}^2\int_0^t
|(\partial_xk_i^n)(\tau,\xi^n(\tau))-(\partial_xk_i^n)(\tau,\xi^m(\tau))|\,d\tau\\
&\hspace{0.5cm}+C_T\sum_{i=0}^2\int_0^t
||(\partial_xk_i^n)(\tau)-(\partial_xk_i^m)(\tau)||_\infty\,d\tau.
\end{align*}
Since $\{\nabla_zf^n\}$ is equicontinuous, so are $\{\partial_xk^n_i\}$, $i=0,1,2$,
due to \eqref{5.14}--\eqref{5.16}. Hence, the second integral in the RHS of the above inequality is $\varepsilon_{nm}$
by Lemma 6.6 $(i)$. Moreover, \eqref{5.14}--\eqref{5.16} implies that the last quantity above
is bounded by $\varepsilon_{nm}$ and $||\nabla_zf^n(\tau)-\nabla_zf^m(\tau)||_\infty$,
and this gives the desired result.
\end{proof}
%
%
\noindent{\bf Convergence of the derivatives.}
We first define the following quantities for simplicity:
\begin{align*}
&f^{nm}_1(t):=||\nabla_zf^n(t)-\nabla_zf^m(t)||_\infty,\\
&\theta_1^{nm}(t):=||\partial_x\alpha^n(t)-\partial_x\alpha^m(t)||_\infty
+||\partial_x\beta^n(t)-\partial_x\beta^m(t)||_\infty.
\end{align*}
Then, Lemma 6.13 and 6.14 are rewritten as follows:
\begin{align*}
&f_1^{(n+1)(m+1)}(t)\leq\varepsilon_{nm}+C_T\int_0^t\theta_1^{nm}(s)\,ds,\\
&\theta_1^{nm}(t)\leq\varepsilon_{nm}+C_T\int_0^t\theta_1^{nm}(s)
+f_1^{nm}(s)\,ds,
\end{align*}
which give the following two integral inequalities:
\begin{align}
&f_1^{(n+1)(m+1)}(t)\leq\varepsilon_{nm}+C_T\int_0^tf_1^{nm}(s)\,ds,\label{6.25}\\
&\theta_1^{nm}(t)\leq\varepsilon_{nm}+C_T\int_0^tf_1^{nm}(s)\,ds.\label{6.26}
\end{align}
We fix the constant $C_T$ in \eqref{6.25}--\eqref{6.26} and
iterate \eqref{6.25} as in the previous results to obtain
\begin{align*}
&\hspace{-0.5cm}f^{(n+1)(m+1)}(t)\\
&\leq\varepsilon_{nm}+C_Tt\varepsilon_{(n-1)(m-1)}
+\frac{(C_Tt)^2}{2}\varepsilon_{(n-2)(m-2)}+\cdots+\frac{(C_Tt)^k}{k!}\varepsilon_{(n-k)(m-k)}\\
&\hspace{0.5cm}+C_T^{k+1}\int_0^t\frac{1}{k!}(t-s)^kf_1^{(n-k)(m-k)}(s)\,ds\\
&\leq e^{C_Tt}\max\{\varepsilon_{nm},\cdots,\varepsilon_{(n-k)(m-k)}\}
+\frac{(C_Tt)^{k+1}}{(k+1)!}||f_1^{(n-k)(m-k)}||_{L^\infty[0,T]}.
\end{align*}
Since $f_1^{(n-k)(m-k)}$ is bounded by a constant, which depends only on $T$, and the quantity
$\frac{(C_Tt)^{k+1}}{(k+1)!}$ converges to zero as $k\rightarrow\infty$, the above estimate
implies that $\{\nabla_zf^n\}$ is Cauchy on $[0,T]$. Finally, \eqref{6.26} implies that
$\{\partial_x\alpha^n\}$ and $\{\partial_x\beta^n\}$ are also Cauchy, and therefore
we conclude that the solution we constructed in Section 6.3.1 is $\mathcal{C}^1$.

\section*{Acknowledgements}
This research has been supported by the TJ Park Science Fellowship of POSCO TJ Park Foundation.
The author would like to thank Prof. Sophonie Blaise Tchapnda for his helpful advice on electromagnetic theory.

\end{document}